\documentclass[12pt]{article}
\usepackage{amsmath}
\usepackage{amssymb}
\usepackage{graphicx}
\usepackage{color}
\graphicspath{{Figures}}
\usepackage[final]{showkeys}
\usepackage[english]{babel}

\definecolor{blue}{rgb}{0,0,.0}
\newcommand{\change}[1]{{\color{blue} #1}}

\newtheorem{theorem}{Theorem}[section]
\newtheorem{definition}[theorem]{Definition}

\newenvironment{proof}
{\noindent
{\bf Proof.}}
{\hfill $\square$\medskip

}

\newcommand{\R}{\mathbb{R}}
\newcommand{\N}{\mathbb{N}}
\newcommand{\bc}{\mathbf{c}}
\newcommand{\bb}{\mathbf{b}}
\newcommand{\ba}{\mathbf{a}}
\newcommand{\br}{\mathbf{r}}
\newcommand{\bbr}{\mathbf{\bar r}}
\newcommand{\bbn}{\mathbf{\bar n}}
\newcommand{\bn}{\mathbf{n}}
\newcommand{\bs}{\mathbf{s}}
\newcommand{\bp}{\mathbf{p}}

\renewcommand{\deg}{\operatorname{deg}}

\newcommand{\supp}{\operatorname{supp}}
\begin{document}
\title{Trimmed Spline Surfaces with\\ Accurate Boundary Control}
\author{Florian Martin and Ulrich Reif}
\date{\today}
\maketitle


\begin{abstract}
We introduce trimmed NURBS surfaces with accurate boundary control,
briefly called ABC-surfaces, as a solution to the notorious
problem of constructing watertight or smooth ($G^1$ and $G^2)$ 
multi-patch surfaces within the function range of standard CAD/CAM 
systems
and the associated file exchange formats.
Our construction is based on the appropriate blend of 
a base surface, which traces out the intended global shape,
and a series of reparametrized ribbons, which dominate
the shape near the boundary.
\end{abstract}

\section{Introduction}
Trimmed NURBS surfaces are the  standard for industrial surface
modeling, today. This class of surfaces combines computational 
efficiency, geometric flexibility, and ease of use for the designer.
All current CAD/CAM systems and various file exchange formats
support this--and typically only this--method.
While the advantages of trimmed NURBS are undisputed, there is 
one severe and persistent shortcoming that has not found a commonly
accepted solution in decades: multi-patch arrangements of trimmed 
NURBS surfaces are typically discontinuous along common boundary 
curves of neighboring pieces. The problem is the following: Given a 
tensor product spline
surface $\bs : \R^2 \to \R^3$, current CAD systems provide two ways 
of \change{representing the boundaries of a trimmed NURBS surface}. 
Either one can use a planar spline curve 
$\gamma : \R \to \R^2$ to bound the domain or one 
can use a spatial spline curve $\bc : \R \to \R^3$ to 
bound the trace. The first variant leads to the boundary curve 
$\bs \circ \gamma$ in three-space. This is a spline curve, 
again, but typically, it has high degree and many knots,
stemming not only from the breaks
of $\gamma$ but also from all points where the trace of $\gamma$
intersects a knot line of $\bs$.
The resulting knot structure is uneven and hardly predictable%
\footnote{
Basically, the only case which does not lead to exploding complexity 
is the special situation where $\gamma$ parametrizes a
straight line parallel to one of the coordinate axes.}.
This fact makes it virtually 
impossible to match $\bs \circ \gamma$ with any other surface 
boundary at reasonable expense. The second variant is equally
discouraging. It seems to be impossible to find a surface 
$\bs$ and a curve $\bc$ such that the trace of $\bs$ contains 
the trace of $\bc$ exactly. The makeshift for this 
problem is to accept disjoint traces, and to
project $\bc$ onto $\bs$ in order to determine the 
actual trimming curve. In general, however, this curve is 
no longer a spline and again not suitable for a watertight
match with a neighboring patch. 

To avoid the described problem, one can try to find arrangements 
of patches for the desired shape which do without trimming.
This approach leads to 
methods like geometric continuity \change{\cite{peters}}, 
singular parametrization \change{\cite{reif_1997}}, or
subdivision \change{\cite{petersreif}}, 
and also to a variety of constructions for $n$-sided 
patches \change{\cite{garcia,varady_2016}}.
As a matter of fact, none of these concepts has found 
widespread acceptance in industrial CAD/CAM applications. 
It may be speculated 
that a certain perseverance of designers and executives
is responsible for that situation. However, there is also a more solid
explanation. The flow of parametric lines is most important for 
the design process%
\footnote{This observation is a blatant contradiction to
the concepts of differentiable geometry, which rely on the 
{\em independence} of geometric properties of the chosen 
parametrization.}.
On one hand, the parametric lines should follow 
the shape characteristics (e.g., ridges or curvature lines) of the 
desired surface to obtain convincing results.
On the other hand, a `natural' arrangement of control meshes 
greatly favors the intuitive character of interactive modeling 
techniques. Thus, there is little flexibility in the choice of
an adequate parametrization, and the corresponding patches must 
be trimmed at places where neighbors meet.
To sum up, the advantages of a well-adapted 
parametrization seem to be rated higher than the notorious 
trouble with the trimming issue.

In CAD/CAM applications, the usual solution to the problem of
nonmatching trimming curves
is to simply accept little gaps in surface representations, and to 
make them so small that they do not impair the quality of the product 
to be manufactured. This allows \change{modeling} surfaces
of adequate quality, but the price to pay is a certain inconvenience 
of the design process. The challenge is to design patches that are 
not only fair by themselves but join with neighbors within a
typically very small margin%
\footnote{Customary tolerances for traditional 
production techniques have to be decreased even 
further for modern 3d printing systems with their incredible 
attention to detail.}. 
Further, not only positions must be
matched, but possibly also surface normals and curvature tensors
in the sense of an approximate $G^1$- or $G^2$-continuity of
the composite surface.
For high quality surface design, as requested for instance in the 
automotive industry, the only way to achieve this is a
time-consuming manual process of iterated boundary adaption and 
fairing. 

Besides these practical considerations, there is another 
issue with the current state of technology. Surface models with 
disconnected pieces are not well suited
for simulation purposes. No matter how small the gaps are, 
they appear insurmountable%
\footnote{Advanced methods like discontinuous Galerkin may 
offer a remedy.}
for heat conduction, transmission
of forces and bending moments, or other physical processes. 
For this reason, multi-patch CAD models are usually converted to 
watertight meshes by elaborate procedures
to prepare them for simulation. Volumetric modeling techniques
as requested for Isogeometric Analysis with its potential
to save this effort are currently studied. However,
they seem to be quite far from a breakthrough for freeform
surface design in industrial applications.

From the abundance of material on the design of surfaces of 
arbitrary 
topology in general and trimmed NURBS in particular, we want 
to briefly mention only two recent approaches.
In \cite{marussig_2017}, a procedure for converting an assembly of 
trimmed 
NURBS patches to a watertight model is suggested.
It is based on local remodeling such that trimming 
curves become parametric lines. Without doubt, this is 
a reasonable practical procedure, even though some issues 
with a special kind of singularities appearing there 
(cf.~\cite{reif_2020})
and also general fairness aspects are still waiting for
a detailed assessment and analysis.
Another line of research
that has to be discussed in this context is documented in 
\cite{varady_2020, varady_2020a}. It is based on ideas similar 
those presented
here, but there are also significant differences. Before
we compare the methods, we briefly explain the basic 
principles of our approach.

\begin{figure}[t]
\begin{center}
\includegraphics[width=.45\textwidth]{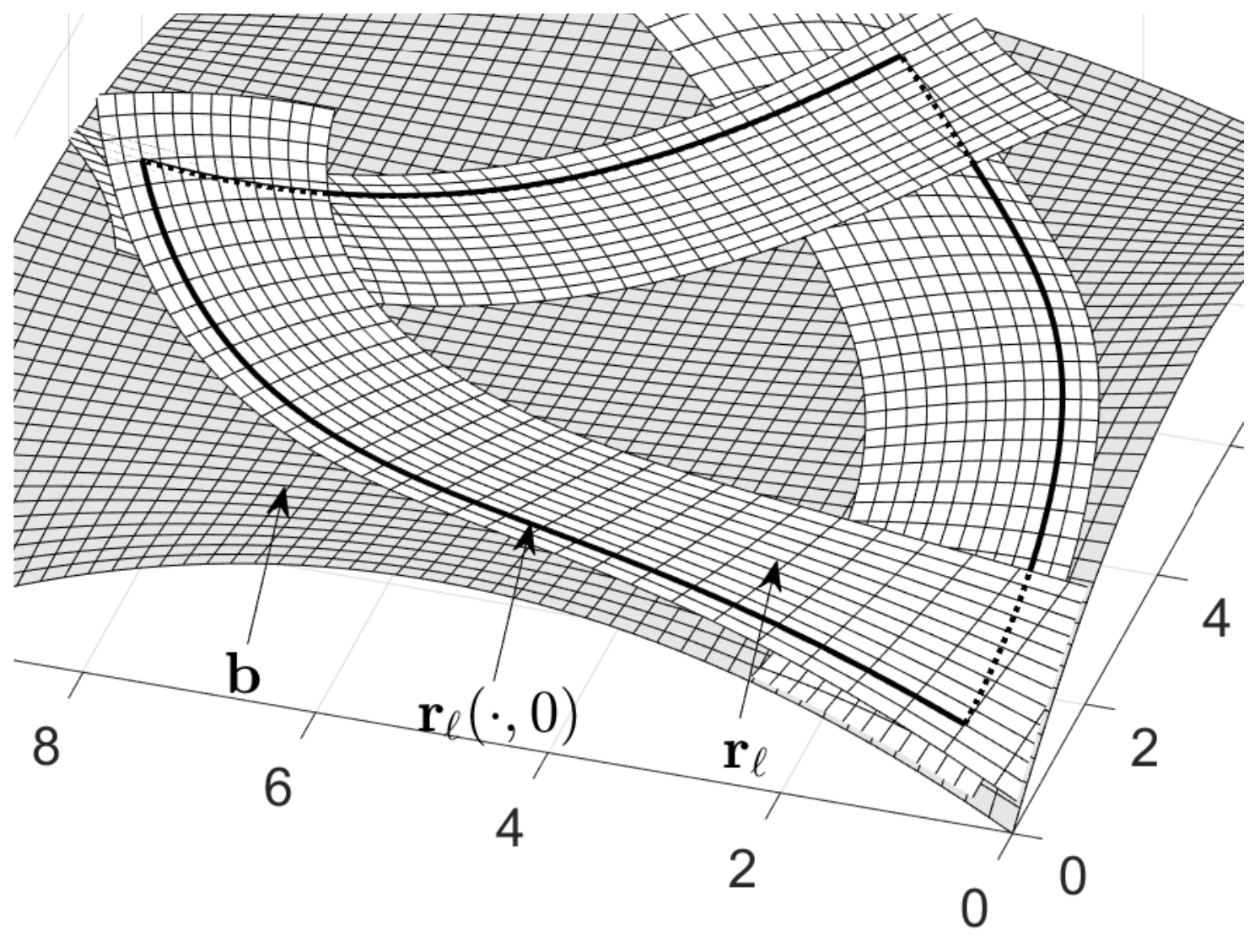}
\qquad 
\includegraphics[width=.45\textwidth]{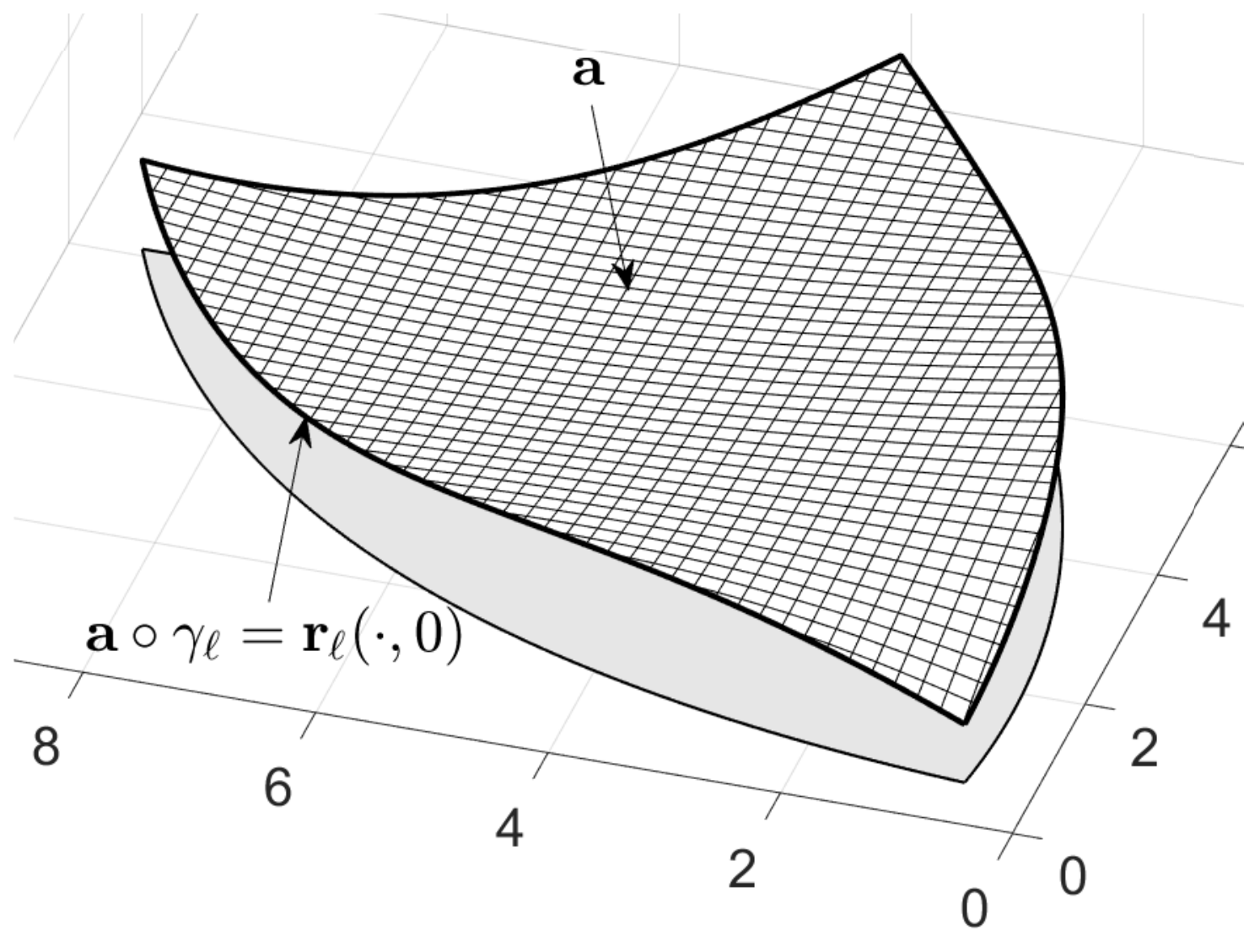}
\end{center}
\caption{Base $\bb$ and ribbons $\br_\ell$ {\it (left)};
ABC-surface $\ba$ {\it (right)}.}
\label{ABCSurface}
\end{figure}
Let $\bb$ be a spline surface,
called {\em base}, and $\br_1,\dots,\br_L$ be a sequence
of other spline surfaces, called {\em ribbons}. Then we 
define the 
{\em ABC-surface}
\[
 \ba := \frac
 {w \cdot \bb + \sum_{\ell=1}^L w_\ell\cdot \br_\ell \circ 
\kappa_\ell}
 {w + \sum_{\ell=1}^L w_\ell}
\]
with suitable reparametrizations $\kappa_\ell$ and weight functions 
$w$ and $w_\ell$. 
Figure~\ref{ABCSurface} illustrates the setting.
\change{It is instructive, but not necessary, to think of $\bb$ as a 
standard spline surface describing the overall shape to 
be designed. When trimmed, its boundaries are close to the desired 
behavior, but do not match exactly. The ribbons $\br_\ell$ are 
employed to adjust that deviation.}

\change{The domain $\Gamma\subset\R^2$ of $\ba$} is assumed to 
be bounded by $L$ smooth segments $\Gamma_1,\dots,\Gamma_L$
meeting at corners $\sigma_1,\dots,\sigma_L$.
The base $\bb$ describes the intended shape globally on $\Gamma$,
and each ribbon $\br_\ell$ describes that shape locally in a vicinity 
of the segment $\Gamma_\ell$ of the boundary. This is done in such a 
way that the 
spline curve $\bc_\ell := \br_\ell(\cdot,0)$ parametrizes the 
segment $\ba(\Gamma_\ell)$ of the boundary of the trace of $\ba$.
The reparametrizations $\kappa_\ell$ are used to reconcile 
the otherwise unrelated surface representations, and the 
weights are used for blending. When approaching the boundary,
$w$ vanishes and makes the base disappear. Further, near the
boundary segment $\Gamma_j \subset \R^2$ corresponding to 
$\bc_j$ all weights $w_\ell$ except
for $w_j$ vanish so that also all ribbons but $\br_j$
are faded out. Hence, this surface dominates the behavior 
of $\ba$ near the boundary in a way that will be made 
precise below. The {\em accurate boundary 
control} of the surface $\ba$ by means of ribbons explains
the acronym {\em ABC-surface}.

The basic principle and some technical details of ABC-surfaces were 
first presented in \cite{reif_2019}. Roughly at the same time 
and independent of our research, another group developed
similar ideas \cite{varady_2020a}, which were recently 
brought to
our attention. The formal similarities are obvious, but there
are also differences, and one of them is significant.

First, the weighted base $w \bb$ does not appear in that
form%
\footnote{It does appear in a similar form in \cite{varady_2012}
for the special case of straight boundaries and polynomial
ribbons.}  in \cite{varady_2020a}.
Instead, a linear combination
$\sum_j b_j \bp_j$ of control points $\bp_j \in \R^3$
and B-splines $b_j$ is used, where only B-splines 
with support completely contained in the domain of the 
surface are permitted so that there is no influence on
the shape near the boundary. Accordingly, the term 
$w$ in the denominator is replaced by $\sum_j b_j$.
This can be considered a minor technical detail, but 
it also reveals a slightly different, though equally valid
philosophy behind the construction. We think of the 
base as a relatively accurate description of the entire surface, 
and of the 
ribbons as small, local correctives, used only to 
adapt $\bb$ to the desired boundary behavior.
By contrast, the ribbons in \cite{varady_2020a} are the major 
building blocks for the shape to be modeled.
The \change{summands} $b_j \bp_j$, if included at all,
are only used for 
variations inside if the shape defined by the ribbons 
alone does not yet meet expectations.

Second, and this difference is much more important,
the surfaces constructed in \cite{varady_2020,varady_2020a} are 
{\em transfinite} in the sense that they do not possess
a closed-form representation by means of elementary
functions. In particular, these surfaces are not piecewise
polynomial or rational and thus cannot be stored in standard file 
formats, as requested in many applications.
The point is that the curves $\gamma_\ell$ bounding the 
domain $\Gamma$ are defined {\em explicitly} as splines.
Once this is done, there is no way to continue the construction
within the spline framework. For instance, it is in 
general impossible to find splines $w_\ell$ that vanish 
on all but one of these boundary curves%
\footnote{
The special case 
of straight boundaries, which admits rational constructions,
was discussed in \cite{varady_2012, varady_2014, varady_2016, 
varady_2017}.}.
In \cite{varady_2020, varady_2020a}, it is suggested to 
construct reparametrizations and 
weights as solutions of variational problems with the curves 
$\gamma_\ell$ providing the boundary conditions. This requests a
certain computational effort, but visually, the results appear to be
excellent.
By contrast, in our construction, the boundary curves $\gamma_\ell$ 
are defined {\em implicitly} as zero levelsets of spline functions.
That way, the problem described above disappears. It may be objected
that explicit evaluation of the boundary of the domain is
now much more expensive than it was before, but as a matter of fact,
this is hardly ever requested in practice as long as the spatial 
boundary curves $\bc_\ell$ are available.

The paper is organized as follows:
In the next section, we introduce and analyze the basic principle for 
the construction of ABC-surfaces in a very general setting. Three
theorems concerning the behavior of positions, normals, and 
curvature tensors at the boundary are provided. They establish
the validity of our construction, but also that of the methods 
presented in the more application-oriented papers 
\cite{varady_2014, varady_2020, varady_2020a}.
In Section~3, we specialize the principle to the spline setting
and describe in detail how to construct adequate reparametrizations 
and weights from given bases and ribbons. 
Section~4 addresses some important technical issues related to the
limitations of current standard formats for file exchange with 
respect to trimming. Finally, in Section~5, we present a few examples 
of ABC-spline surfaces to demonstrate the potential of our method.

\section{Basic principle}
Throughout, boldface letters denote points and functions with values 
in $\R^3$, while greek letters denote points and functions with 
values in $\R^2$.
Except for the surface $\ba$ to be constructed,
all functions appearing below are assumed to be sufficiently
smooth. In particular, their derivatives
mentioned anywhere (explicitly or implicitly through notions 
like normals or curvatures) are assumed to exist and to be continuous.
Further, to avoid a discussion of degenerate cases, surfaces are 
always assumed to be immersed, i.e., the partial derivatives are 
linearly independent everywhere.
The following definitions are crucial for the construction of
ABC-surfaces:

\begin{figure}[t]
\begin{center}
\includegraphics[width=.45\textwidth]{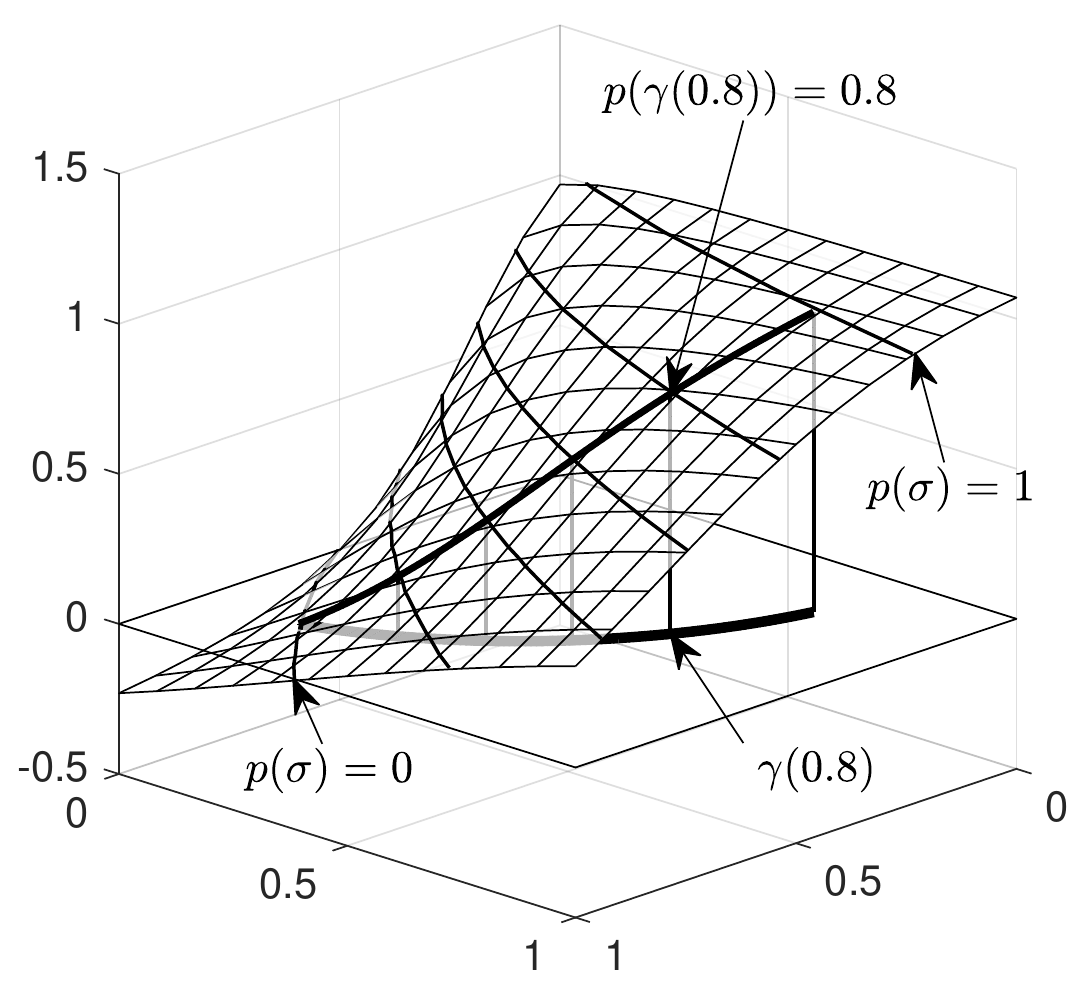}
\qquad 
\includegraphics[width=.45\textwidth]{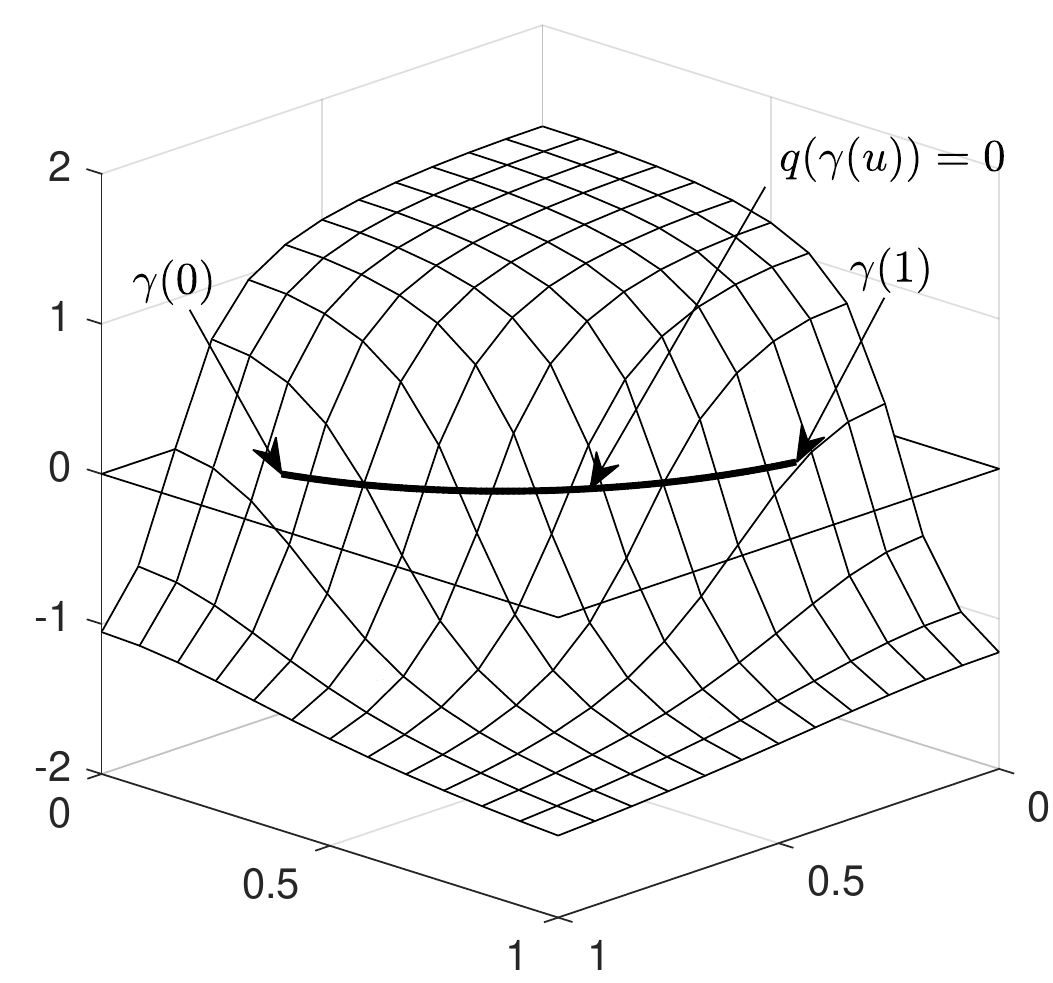}
\end{center}
\caption{Parametrizing function $p$ {\it (left)}
and implicit representation $q$ {\it (right)}.}
\label{Para+Implicit}
\end{figure}
\begin{definition}
\label{def:kappa}
The function $p : \R^2 \to \R$ {\em parametrizes} the planar 
curve $\gamma : [0,1] \to \R^2$ if $p(\gamma(u)) = u$ for all $u \in 
[0,1]$, see Figure~\ref{Para+Implicit} {\it (left)}.
The function $q: \R^2 \to \R$ is an {\em implicit 
representation} of the planar curve $\gamma$ if $q(\gamma(u)) = 0$ 
for all 
$u \in [0,1]$, see Figure~\ref{Para+Implicit} {\it (right)}.
The function
$\kappa := [p,q] : \R^2 \to \R^2$ {\em defines} $\gamma$ if there 
exists a unique regular curve $\gamma : [0,1] \to \R^2$ with
implicit representation $q$ that is parametrized by $p$.
\end{definition}
Now, we describe the construction of a surface $\ba$, whose domain 
$\Gamma \subset\R^2$ of 
definition is bounded by a sequence $\gamma_1,\dots,\gamma_L$ of 
planar {\em trimming curves} $\gamma_\ell : [0,1] \to \R^2$.
When concatenated according to
\begin{equation}
\label{eq:gamma_loop}
 \gamma_{\ell-1}(1) = \gamma_\ell(0)
 ,
\end{equation}
these curves form a closed loop without self-intersections.
Here and below, the index $\ell$ always runs from $1$ through $L$ and 
is understood modulo $L$. Let
\[
 \sigma_\ell := \gamma_\ell(0)
 \quad\text{and}\quad 
 \Gamma_\ell := \{\gamma_\ell(u) : 0 \le u <1\}
 ,
\]
denote the {\em corners} and {\em boundary segments}
of $\Gamma$, respectively. Formally, the endpoint $\gamma_\ell(1)$ 
is 
missing in the definition of $\Gamma_\ell$ to obtain a disjoint
union $\partial \Gamma = \bigcup_\ell \Gamma_\ell$ for the boundary 
of $\Gamma$. Accordingly, for each boundary point $\beta \in \partial 
\Gamma$, there exist unique values $\ell$ and 
$u \in [0,1)$ such $\beta = \gamma_\ell(u)$.

\begin{figure}[t]
\begin{center}
\includegraphics[width=.45\textwidth]{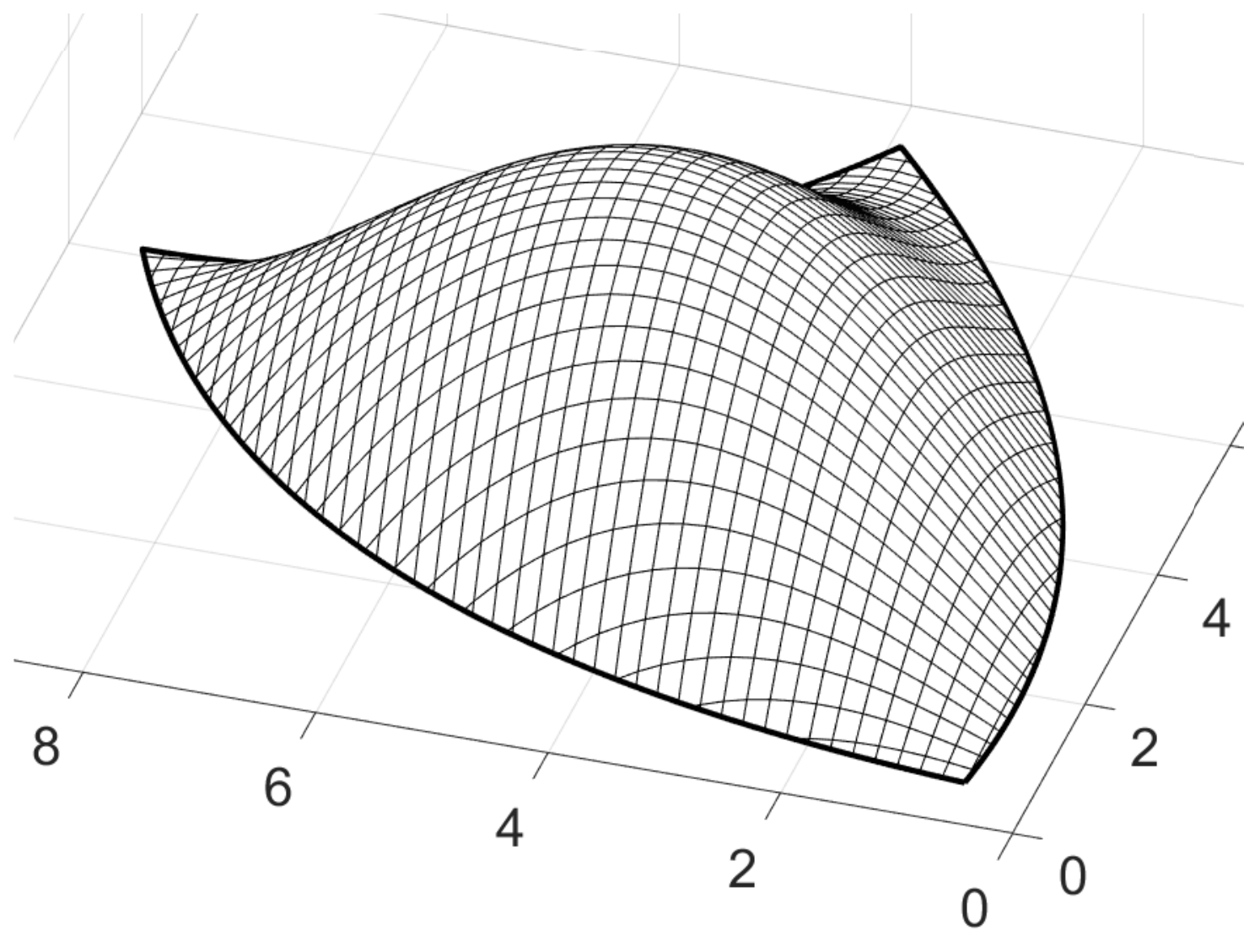}
\qquad 
\includegraphics[width=.45\textwidth]{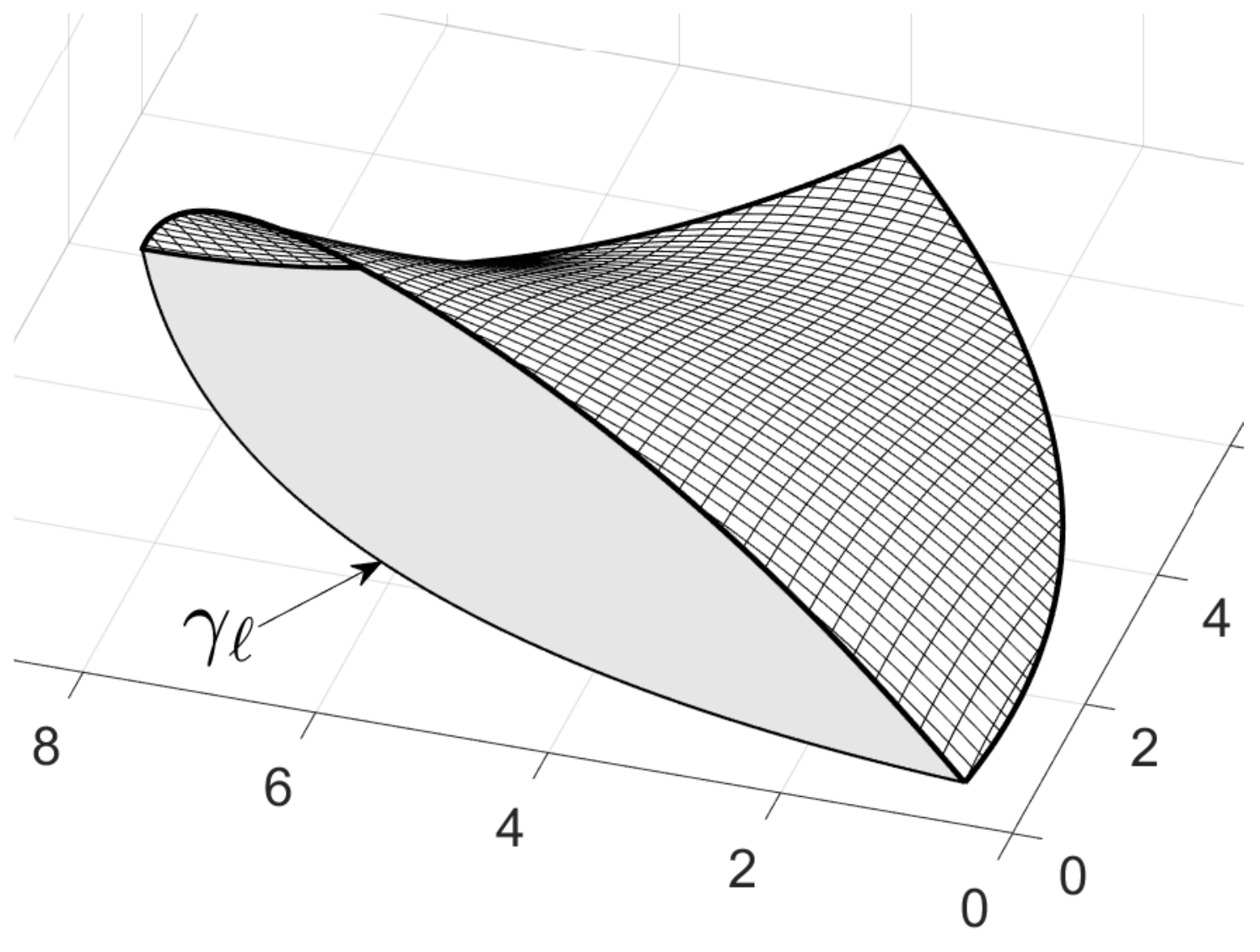}
\end{center}
\caption{Weights $w$ {\it (left)} and $w_\ell$ {\it (right)}.}
\label{Weights}
\end{figure}

The building blocks for the construction of the surface $\ba$ are the 
following:
\begin{itemize}
\item
Let $\kappa_\ell := [p_\ell,q_\ell] : \R^2 \to 
\R^2$ be functions that define trimming curves $\gamma_\ell$ 
in the sense of Definition~\ref{def:kappa}
according to the setting described above.
\item 
Let $w : \R^2 \to \R$ be an implicit 
representation of all trimming curves. That is,
\[
 w\circ \gamma_j = 0
 ,\quad 
 j = 1,\dots,L
 .
\]
Further, it is assumed
that $w$ is positive on the interior of $\Gamma$,
see Figure~\ref{Weights} {\it (left)}.
\item
For each $\ell$,
let $w_\ell:\R^2 \to \R^2$ 
be an implicit representation of all  trimming curves but 
$\gamma_\ell$.
That is,
\[
 w_\ell\circ\gamma_j = 0
 ,\quad 
 j = 1,\dots,\ell-1,\ell+1,\dots,L
 .
\]
Further, it is assumed that $w_\ell$ is non-negative on $\Gamma$, and 
that
the only zeros of the sum $\sum_\ell w_\ell$ in $\Gamma$ are the 
corners,
see Figure~\ref{Weights} {\it (right)}.
\item
Let $\bb : \R^2 \to \R^3$ be a surface, called 
the {\em base}.
Unit normal
and curvature tensor of $\bb$ are denoted by
$\bn_\bb$ and $E_\bb$, respectively.
\item 
For $\ell=1,\dots,L$, let
$\br_\ell : \R^2 \to \R^3$ be surfaces, called 
{\em ribbons}. Unit normal
and curvature tensor of $\br_\ell$ are denoted by
$\bn_\ell$ and $E_\ell$, respectively. 
\end{itemize}
With these ingredients, we construct a new surface that 
combines base and ribbons in a specific way:

\begin{definition}
\label{def:ABC}
The {\em ABC-surface} $\ba : \Gamma \to \R^3$ corresponding to 
the data $\bb,\br_\ell,w,w_\ell,\kappa_\ell$, as defined above,
is given by
\begin{equation}
\label{eq:ABC}
 \ba := \frac
 {w \cdot \bb + \sum_{\ell=1}^L w_\ell\cdot \br_\ell \circ 
\kappa_\ell}
 {w + \sum_{\ell=1}^L w_\ell}.
\end{equation}
At corners, where the right hand
side is not yet well defined, let
\[
 \ba(\sigma_\ell) := \br_\ell(0,0)
 ,\quad 
 \ell = 1,\dots,L
 .
\]
\end{definition}
For the sake of simplicity, 
this definition is restricted to the case that the domain $\Gamma$
is simply connected and has always $L$ corners and boundary segments.
Generalizations to multiply-connected domains and to 
disc-like surfaces bounded by one smooth, periodic
curve will be discussed in a forthcoming report.

Two basic elements characterize the construction in the definition 
above. First, 
the functions $\kappa_\ell$, which define the trimming
curves, are used as {\em repara\-metri\-zations} for the ribbons 
$\br_\ell$ to 
adapt them to the otherwise unrelated base $\bb$. 
Second, near the 
boundary segment $\Gamma_j$ the \change{various} implicit 
representations \change{$w$ and $w_\ell$}
are used as {\em weights} to fade out the influence of
the base $\bb$ and all ribbons $\br_\ell$ 
except for $\br_j$. The ribbon $\br_j$ becomes
dominant and determines the shape of $\ba$. Thus, the construction
provides
{\em accurate boundary control} for $\ba$, giving rise to the acronym
{\em ABC-surface}. For a more detailed assessment of this raw
idea, we need a further definition:
\begin{definition}
\label{def:contact}
Let $\ba$ be an ABC-surface according to Definition~\ref{def:ABC}.
It has {\em contact order $k$} at the boundary segment $\Gamma_j$ if
for any $\sigma = \gamma_j(u) \in \Gamma_j$ it holds
\[
 \bar w(\tau) := \frac{w(\sigma+\tau)}{w_j(\sigma+\tau)} = 
 O(|\tau|^{k+1})
\]
and, unless $\ell  = j$ or $\sigma=\sigma_j$ and $\ell=j-1$, also
\[
 \bar w_\ell(\tau) := \frac{w_\ell(\sigma+\tau)}{w_j(\sigma+\tau)} = 
 O(|\tau|^{k+1})
 .
\]
\end{definition}
Further introducing the notation
\[
  \bbr_\ell(\tau) := \br_\ell(\kappa_\ell(\sigma+\tau))
\] 
for the reparametrized ribbons, $\ba$ can be written in the form 
\begin{equation}
\label{eq:atau}
 \ba(\sigma+\tau) = 
 \frac{\bar w(\tau)\bb(\sigma+\tau) + \sum_\ell \bar 
w_\ell(\tau)\bbr_\ell(\tau)}
 {\bar w(\tau) + \sum_\ell \bar w_\ell(\tau)}
 .
\end{equation}
Here and below, the variable $\sigma$ is always assumed to be fixed
and mostly omitted in the notation for the sake of simplicity.
Unit normal $\bbn_\ell$ and 
curvature tensor $\bar E_\ell$ of $\bbr_\ell$ coincide with 
the respective quantities of $\br_\ell$ since these surfaces
are related by reparametrization,
\begin{equation}
\label{eq:n=n}
 \bbn_\ell(\tau) = \bn_\ell(\kappa_\ell(\sigma+\tau))
 ,\quad 
 \bar E_\ell(\tau) = E_\ell(\kappa_\ell(\sigma+\tau))
 .
\end{equation}
In the interior of $\Gamma$, $\ba$ is a blend of smooth surfaces and 
thus can be expected to be smooth, too. As always, exceptions may 
occur in case of an incidental loss of regularity. We will not 
discuss such degenerate cases, but focus on the behavior of $\ba$ 
at the boundary. The following three theorems clarify the situation:
\begin{theorem}[$G^0$-boundary]
\label{thm:G0}
Consider an ABC-surface $\ba$ according to Definition~\ref{def:ABC}
with contact order $k=0$ at $\Gamma_j$.
If the ribbons are consistent according to 
\[
 \br_{j-1}(1,0) = \br_j(0,0)
 ,
\]
then $\ba$ is continuous 
at $\Gamma_j$, and it holds 
\begin{equation}
\label{eq:G0}
 \ba(\gamma_j(u)) = \br_j(u,0)
 ,\quad 
 u \in [0,1]
 .
\end{equation}
\end{theorem}
In the following proofs, it is convenient to use the 
abbreviations
\[
 h(\tau) := \frac{1}{1+\bar w_{j-1}(\tau)}
 =
 \frac{w_j(\sigma_j+\tau)}{w_{j-1}(\sigma_j+\tau)+w_j(\sigma_j+\tau)}
 ,\quad 
 h'(\tau) := 1-h(\tau)
 ,
\]
to expand $\ba$ in a vicinity of the corner $\sigma_j$.
Since weights are assumed to be nonnegative, it 
is $h(\tau) \in [0,1]$.
For contact order $k = r-1$, we obtain
\begin{align}
\notag
 \ba(\sigma_j+\tau) &= 
 \frac{\bbr_j(\tau) + \bar w_{j-1}(\tau) 
 \bbr_{j-1}(\tau) + O(|\tau|^r)}{1+\bar w_{j-1}(\tau) + O(|\tau|^r)}\\
 \label{eq:h}
 &=
 h(\tau)\bbr_j(\tau) + h'(\tau)\bbr_{j-1}(\tau) + O(|\tau|^r)
 .
\end{align}

\begin{proof}
Consider the point $\sigma = \gamma_j(u) \in \Gamma_j$
and $\ba$ in the form \eqref{eq:atau}.
We note that $\bbr_{j}(\tau) = \br_j(u,0) + O(|\tau|)$ and,
by consistency, $\bbr_{j-1}(\tau) = \br_j(0,0) + O(|\tau|)$.
For $u >0$, \eqref{eq:G0} follows from
\[
 \ba(\sigma+\tau) = 
 \frac{\bbr_j(\tau) + O(|\tau|)}{1+O(|\tau|)}
 = \bbr_j(\tau) + O(|\tau|)
 = \br_j(u,0) + O(|\tau|)
 .
\]
For $u=0$, \eqref{eq:h} yields
\[
 \ba(\sigma+\tau) = 
 h(\tau)\bbr_j(\tau) + h'(\tau)\bbr_{j-1}(\tau) + O(|\tau|)
 = \br_j(0,0) + O(|\tau|)
 .
\]
\end{proof}
In the following, $\partial_1 f,\partial_2 f$ denote the partial 
derivatives of the bivariate function $f : \R^2 \to \R^n$, 
respectively.
$Df := [\partial_1 f\ \partial_2 f] : \R^2 \to \R^{n \times 2}$ 
denotes the Jacobian.
\begin{theorem}[$G^1$-boundary]
\label{thm:G1}
Consider an ABC-surface $\ba$ according to Definition~\ref{def:ABC}
with contact order $k=1$ at $\Gamma_j$.
If the conditions of Theorem~\ref{thm:G0} are satisfied, if
the unit normals of ribbons are consistent according to 
\[
 \bn_{j-1}(1,0) = \bn_j(0,0)
 ,
\]
and if the regularity condition
\begin{equation}
 \label{eq:regular}
 \operatorname{rank}
 \bigl(
 t D\bbr_{j}(0) + (1 - t) D\bbr_{j-1}(0)\bigr) = 2
\end{equation}
is satisfied for all $t \in [0,1]$, then
$\ba$ has a well-defined unit normal $\bn_\ba$
at $\Gamma_j$, and it holds 
\begin{equation}
\label{eq:G1}
 \bn_\ba(\gamma_j(u)) = \bn_j(u,0)
 ,\quad 
 u \in [0,1]
 .
\end{equation}
\end{theorem}
\begin{proof}
Analogous to the previous proof, we obtain for $u>0$
\[
 \ba(\sigma+\tau) = 
 \frac{\bbr_j(\tau) + O(|\tau|^2)}{1+O(|\tau|^2)}
 = \bbr_j(\tau) + O(|\tau|^2)
 .
\]
Differing only by second order terms, the surfaces
$\ba(\sigma+\cdot)$ and $\bbr$ have the same normal at the origin.
Hence, using \eqref{eq:n=n}, we obtain
\[
 \bn_\ba(\sigma) = \bbn_j(0) = \bn_j(u,0)
 .
\]
Below, we mostly omit not only \change{notating} $\sigma$, 
but also the argument $\tau$.
For $u=0$, corresponding to the case $\sigma=\sigma_j$,
\eqref{eq:h} yields the derivative
\[
 D\ba(\sigma_j+\tau) = 
 h D\bbr_j + h' D\bbr_{j-1}
 + (\bbr_j-\bbr_{j-1}) Dh + O(|\tau|)
 .
\]
By consistency of ribbons, it is 
$\bbr_j-\bbr_{j-1} = O(|\tau|)$ so that 
\[
 D\ba(\sigma_j+\tau) = 
 h D\bbr_j(0) + h' D\bbr_{j-1}(0)
 + O(|\tau|)
 .
\]
Let us consider the leading matrix, defined by the first two 
summands. Since $h \in [0,1]$,
it is a convex combination of $D\bbr_j(0)$ and $D\bbr_{j-1}(0)$.
By the regularity condition, it has rank $2$, and by consistency of
normal vectors, both columns are perpendicular to 
$\bn_j(0,0)$. Hence, normalizing the cross product of these columns
and passing to the limit $\tau \to 0$
yields $\bn_\ba(\sigma_j) = \bn_j(0,0)$, as stated.
Clearly, in general, the derivative $D\ba(\sigma_j+\tau)$ itself
does not possess a limit for $\tau \to 0$.
\end{proof}
We note that the regularity condition \eqref{eq:regular} is purely 
technical and typically satisfied. It 
rules out a degenerate situation, which can occur in a similar way 
anywhere in any common surface modeling system for an unfortune 
choice of control data. 

The crucial conditions for the
existence of normal vectors at the boundary are {\em algebraic} 
(proper decay of the functions $\bar w$ and $\bar w_\ell$) and 
{\em geometric}
(consistency of ribbons and normals). Similar conditions appear in 
the following
theorem concerning curvature tensors. However, and this may 
be unexpected, also a {\em parametric} condition (consistency of 
derivatives of ribbons) has to be met to guarantee existence of 
a curvature tensor at corners.
\begin{theorem}[$G^2$-boundary]
\label{thm:G2}
Consider an ABC-surface $\ba$ according to Definition~\ref{def:ABC}
with contact order $k=2$ at $\Gamma_j$.
If the conditions of theorems~\ref{thm:G0} and \ref{thm:G1} are 
satisfied, if
the curvature tensors of ribbons are consistent according to 
\[
 E_{j-1}(1,0) = E_j(0,0)
 ,
\]
and if the condition
\begin{equation}
 \label{eq:Dr=Dr}
 D\bbr_{j-1}(0) = D\bbr_j(0)
\end{equation}
is satisfied, then
$\ba$ has a well-defined curvature tensor $E_\ba$
at $\Gamma_j$, and it holds 
\begin{equation}
\label{eq:G2}
 E_\ba(\gamma_j(u)) = E_j(u,0)
 ,\quad 
 u \in [0,1]
 .
\end{equation}
\end{theorem}
\begin{proof}
As before, we obtain for $u>0$
\[
 \ba(\sigma+\tau) = 
 \frac{\bbr_j(\tau) + O(|\tau|^3)}{1+O(|\tau|^3)}
 = \bbr_j(\tau) + O(|\tau|^3)
 .
\]
Differing only by third order terms, the surfaces
$\ba(\sigma+\cdot)$ and $\bbr$ have the same curvature tensor at the 
origin. Hence, using \eqref{eq:n=n}, we obtain
\[
 E_\ba(\sigma) = E_j(0) = E_j(u,0)
 .
\]
Analogous to the previous proof, we obtain for 
$u=0$, corresponding to the case $\sigma=\sigma_j$,
\[
 \ba(\sigma_j+\tau)
 =
 h \bbr_j + h' \bbr_{j-1} + O(|\tau|^3)
\]
and, using \eqref{eq:Dr=Dr},
\[
 D\ba(\sigma_j+\tau) = 
 h D\bbr_j(0) + h' D\bbr_{j-1}(0)
 + O(|\tau|)
 =
 D\bbr_j(0) + O(|\tau|)
 .
\]
This time, the derivative has a limit at the corner,
\begin{equation}
 \label{eq:Da=Dr}
 D\ba(\sigma_j) = D\bbr_j(0)
 .
\end{equation}
The second order partial derivatives are given by
\begin{align*}
 \partial_{\mu\nu} \ba(\sigma_j+\tau) =
 &\partial_{\mu\nu}h(\bbr_j-\bbr_{j-1}) +
 \partial_{\mu} h\, (\partial_\nu\bbr_j - \partial_\nu\bbr_{j-1} )+\\
 &\partial_{\nu} h\, (\partial_\mu\bbr_j - \partial_\mu\bbr_{j-1} )+
 h \partial_{\mu\nu}\bbr_j +h' \partial_{\mu\nu}\bbr_{j-1} + 
O(|\tau|)\\
= 
 &h \partial_{\mu\nu}\bbr_j(0) +h' \partial_{\mu\nu}\bbr_{j-1}(0) 
 + O(|\tau|),
\end{align*}
where we used that $\bbr_{j-1}(0) = \bbr_j(0)$ and 
$D\bbr_{j-1}(0) = D\bbr_j(0)$.
Further, the curvature tensors of $\bbr_{j-1}$ and
$\bbr_j$ coincide at the origin, $\bar E_{j-1}(0) = \bar E_j(0)$.
This implies that also the second fundamental forms coincide,
\[
 \bar B_{j-1}^{\mu\nu}(0) := 
 \langle \bbn_{j-1}(0), \partial_{\mu\nu}\bbr_{j-1}(0) \rangle
 =
 \langle \bbn_{j}(0), \partial_{\mu\nu}\bbr_{j}(0) \rangle
 =: 
 \bar B_{j}^{\mu\nu}(0)
 .
\]
The second fundamental form of $\ba$ is
\[
 B_\ba^{\mu\nu}(\sigma_j+\tau) := 
 \langle \bn_\ba(\sigma_j+\tau), 
 \partial_{\mu\nu}\ba(\sigma_j+\tau) \rangle 
 .
\]
By Theorem~\ref{thm:G1},
the normal vector of $\ba$ converges, i.e.,
$\bn_\ba(\sigma_j+\tau) = \bn_\ba(\sigma_j) + O(|\tau|)$.
Hence, we obtain
\begin{align*}
 B_\ba^{\mu\nu}(\sigma_j+\tau) &=
 h \langle 
\bn_\ba(\sigma_j),\partial_{\mu\nu}\bbr_{j}(0) \rangle+
h' \langle
\bn_\ba(\sigma_j),\partial_{\mu\nu}\bbr_{j-1}(0) \rangle
+ O(|\tau|)\\
&=
h B_j^{\mu\nu}(0) + h' B_{j-1}^{\mu\nu}(0) + O(|\tau|)
= B_j^{\mu\nu}(0) + O(|\tau|)
\end{align*}
and finally, using $\bn_\ba(\sigma_j) = \bbn_{j}(0)
= \bbn_{j-1}(0)$,
\[
 B_\ba^{\mu\nu}(\sigma_j) = B_j^{\mu\nu}(0)
 .
\]
Together with \eqref{eq:Da=Dr} and \eqref{eq:n=n}, this implies 
coincidence of curvature tensors, 
\[
 E_\ba(\sigma_j) = \bar E_j(0) = E_j(0,0)
 ,
\]
and the proof is complete.
\end{proof}
The proofs of the three theorems show that $G^k$-properties
on the interior part of the boundary segments $(0 < u < 1)$
are easily achieved by requesting proper decay of the functions 
$\bar w$ and $\bar w_\ell$. Consistency and regularity conditions
are requested only to guarantee good behavior at the corners.
Remarkably, condition \eqref{eq:Dr=Dr} concerning 
consistency of the derivatives of reparametrized ribbons shows that
the functions $\kappa_\ell$ cannot be chosen independently of the 
ribbons $\br_\ell$. More precisely, the interrelation
\begin{equation}
\label{eq:Dbbr}
 D\bbr_{j-1}(0) = D\br_{j-1}(1,0) D\kappa_{j-1}(\sigma_j) = 
 D\br_{j}(0,0) D\kappa_{j}(\sigma_j) = D\bbr_{j}(0)
\end{equation}
must be observed.
Assuming that the ribbons are given and the reparametrizations 
are still unknown, this equation yields a linear constraint
on the pair $D\kappa_{j-1},D\kappa_j$ at the corner $\sigma_j$.
Vice versa, if the reparametrizations are given, the space of 
admissible ribbons is restricted to a linear subspace.

The following example shows that this condition is in fact crucial.
Let 
\[
 \bbr_{j-1}(x,y) = 
 \begin{bmatrix}
 x \\ y \\ x^2                
 \end{bmatrix}
 ,\quad 
 \bbr_{j}(x,y) = 
 \begin{bmatrix}
 2x \\ 2y \\ 4x^2                
 \end{bmatrix} 
 ,\quad 
 (x,y) \in [0,1]^2
 .
\]
Near the origin, the traces of both surfaces coincide so that 
consistency of point, normal, and curvature tensor is guaranteed.
However, 
\[
 D\bbr_{j-1}(0,0) =
 \begin{bmatrix}
 1 & 0\\  0 & 1 \\ 0 & 0
 \end{bmatrix}
 \neq
 \begin{bmatrix}
 2 & 0\\  0 & 2 \\ 0 & 0
 \end{bmatrix}
 = D\bbr_j(0,0)
 .
\]
The weight functions
$w_{j-1}(x,y) = x^3, w_j(x,y) = y^3$ have proper decay at the 
boundary curves $\gamma_{j-1}(u) = [u,0,0]$ and 
$\gamma_j(u) = [0,u,0]$. Assuming for simplicity that 
the base and all other ribbons vanish, we obtain
\[
 \ba(x,y) = \frac{1}{x^3+y^3}\,
 \begin{bmatrix}
 x(x^3+2y^3) \\y(x^3+2y^3)\\x^2(x^3+4y^3)
 \end{bmatrix}
 .
\]
Skipping the technical details, we note that 
\[
 \lim_{(x,y)\to (0,0)} \ba(x,y) = 
 \bbr_{j-1}(0,0) =\bbr_j(0,0) =
 \begin{bmatrix}
 0\\0\\0 
 \end{bmatrix}
\]
and
\[
 \lim_{(x,y)\to (0,0)} \bn_\ba(x,y) = 
 \bbn_{j-1}(0,0) =\bbn_j(0,0) =
 \begin{bmatrix}
 0\\0\\1 
 \end{bmatrix}
 .
\]
However, the curvature tensor $E_\ba$ is discontinuous at the origin.
Figure~\ref{counter_example} shows the surface $\ba$ 
\change{{\it (left)}} and the 
Gaussian curvature 
$K(r\cos \varphi,r\sin \varphi)$ \change{{\it (right)}
as a function of $r$.
The individual lines correspond to $30$ equally spaced angles 
$\varphi$
in the interval $[0,\pi/2]$}. For $\varphi=0$
and $\varphi=\pi/2$, corresponding to the trimming curves 
$\gamma_{j-1}$ and $\gamma_j$, \change{it is 
$\lim_{r\to 0} K(r,0) = \lim_{r=0}K(0,r)=0$},
as suggested by 
the cylindrical shape of $\bbr_{j-1}$ and $\bbr_j$. However,
other angles yield non-zero limits.
\begin{figure}[t]
\begin{center}
\includegraphics[width=.49\textwidth]{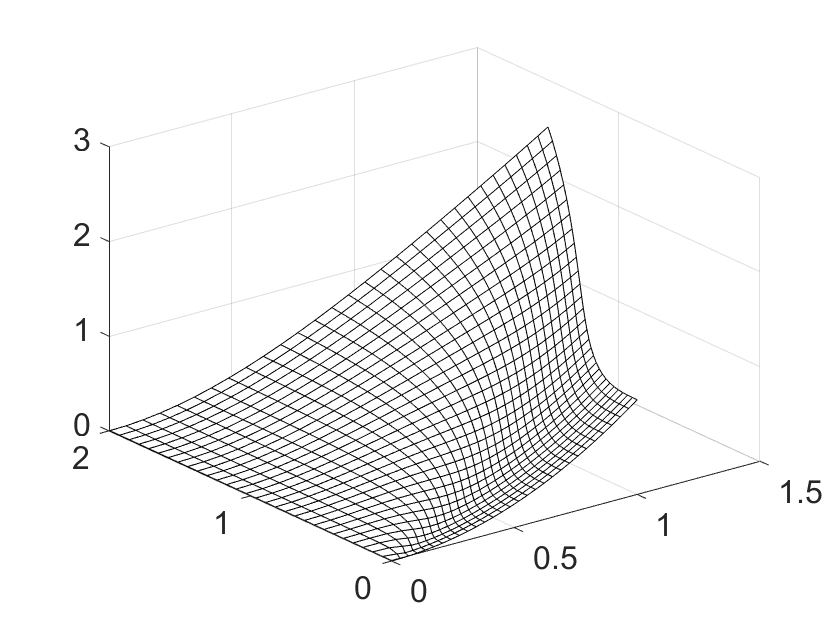}
\includegraphics[width=.49\textwidth]{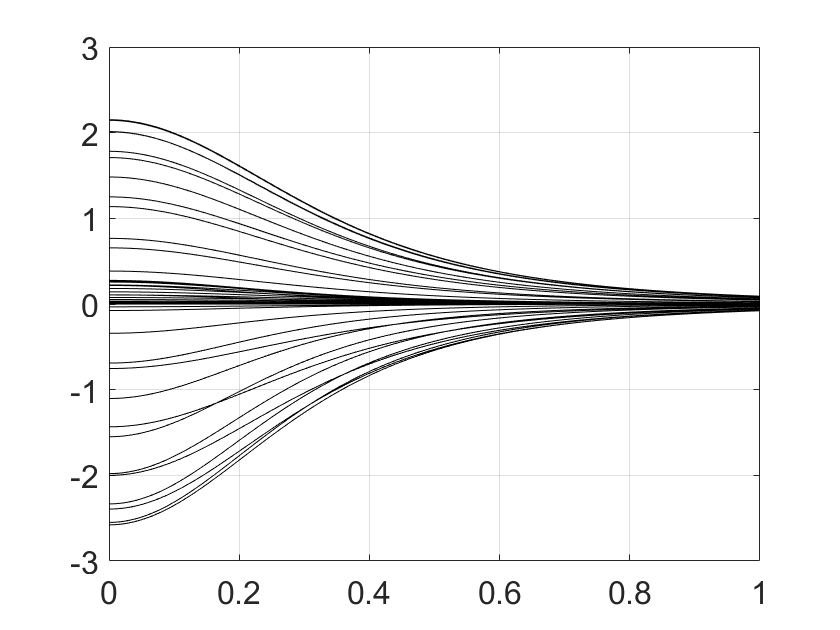}
\end{center}
\caption{\change{Surface $\ba(x,y)$ obtained by 
blending} inappropriately repara\-metrized ribbons {\it 
(left)} and \change{its} Gaussian curvature
$K(r\cos \varphi, r\sin \varphi)$
\change{plottet over $r \in (0,1]$}
for \change{30} different angles $\change{\varphi \in [0,\pi/2]}$
{\it (right)}.
}
\label{counter_example}
\end{figure}
%


\section{ABC-spline surfaces}

In the previous section, the ABC-principle was introduced in great 
generality. Our main goal, however, is to find a special embodiment 
of the setup that is fully compatible with the predominant standard 
of industrial surface modeling: trimmed NURBS.

To this end, we assume that all building blocks of the ABC-surface 
$\ba$ are tensor product NURBS. Even more, for the sake 
of simplicity, we consider only purely polynomial, i.e., non-rational 
splines.
Further, it is \change{practical} (but not necessary) to use 
a common set 
of knot vectors for the base $\bb$, the weights $w,w_\ell$
and the reparametrizations $\kappa_\ell$; degrees may be different.
By contrast, the knots of the ribbons $\br_\ell$ are completely 
independent. It is immediately clear that the concatenation, 
multiplication, summation and division of piecewise polynomial
functions yields a piecewise rational function $\ba$.
Denoting the coordinate degree of the bivariate spline $s$ by 
$\deg s\in \N_0^2$ and its total degree by $|\deg s| \in \N_0$, 
we obtain 
\begin{align*}
 d &:= \deg (w \bb) = \deg w + \deg \bb\\
 d_\ell &:= \deg (w_\ell \change{\br_\ell} \circ \kappa_\ell)
 = \deg w_\ell + |\deg \br_\ell| \deg \kappa_\ell
 .
\end{align*}
Together, the maximal degree of the rational pieces of $\ba$ equals 
that of the corresponding numerators, and is given by
\begin{equation}
\label{eq:deg}
 \deg \ba = \max\{d,d_1,\dots,d_L\}
 ,
\end{equation}
where the maximum is taken component-wise.
More specific results will be presented later on.

The crucial question that has to be addressed now concerns the 
construction of suitable splines $w,w_\ell$, and $\kappa_\ell$. 
In recent works \cite{varady_2020, varady_2020a}, reparametrizations 
for domains bounded by
curves are obtained as solutions to variational problems with the 
given trimming curves $\gamma_\ell$ as boundary conditions.
It can be expected that this yields good visual results, but 
of course, there does not exist a closed-form solution, and much less 
a spline fulfilling the condition $\kappa_\ell(\gamma(u)) = [u,0]$.
Similar problems arise when seeking weights vanishing at the boundary
or parts thereof.

The solution of the problem to find weights and reparametrizations 
compatible with the boundary is amazingly simple: There is no need to 
define the boundary curves $\gamma_\ell$ explicitly. It suffices
to define them {\em implicitly} by the functions $\kappa_\ell$ 
according to Definition~\ref{def:kappa},
\[
    \gamma_\ell(u) := \{\xi \in \R^2 : \kappa_\ell(\xi) = [u, 0]\}
    .
\]
This setting is justified by the fact that standard CAD systems and 
also standard file exchange formats like STEP or IGES allow 
trimming curves to be specified explicitly as a closed loop of spline 
curves in the parametric domain, or alternatively--and this is the 
option used here--as a closed loop of spline curves in $\R^3$
bounding the trace of the given NURBS surface. By 
Theorem~\ref{thm:G0}, the curves
\[
 [0,1] \ni u \ \mapsto\ 
 \bc_\ell(u) := \br_\ell(u,0) = \ba(\gamma_\ell(u))
 ,\quad 
 \ell = 1,\dots,L
 ,
\]
can serve exactly that purpose. Actually, the planar trimming curves 
$\gamma_\ell$ need not be computed at any processing step of 
ABC-spline surfaces as long as geometric trimming by spatial curves,
as explaind in Section~4, is available.

Now, we elaborate on methods how to determine suitable 
reparametrizations $\kappa_\ell$ and weights $w,w_\ell$
under the assumption that the base $\bb$ and the ribbons $\br_\ell$ 
are given. 

\subsection{Reparametrization}

In principle, the functions $\kappa_\ell$ can be chosen quite
arbitrarily. Basically, only the consistency condition
\eqref{eq:gamma_loop} must be satisfied. This means that there 
exist points $\sigma_\ell\in \R^2$ such that 
\[
 \kappa_\ell(\sigma_\ell) = [0,0]
 ,\quad 
 \kappa_{\ell-1}(\sigma_\ell) = [1,0]
\]
for all $\ell$. If continuity of curvature tensors is 
requested, also condition \eqref{eq:Dbbr} must be observed, i.e.,
\begin{equation}
\label{eq:con1}
  D\br_{\ell-1}(1,0)\cdot
  D\kappa_{\ell-1}(\sigma_\ell) = 
  D\br_{\ell}(0,0)\cdot D\kappa_\ell(\sigma_\ell)
  .
\end{equation}
\change{Further}, to obtain fair surfaces, it is reasonable to 
request that 
$\kappa_\ell$ repara\-metrizes $\br_\ell$ in such a way that 
it is close to the base,
\[
 \br_\ell \circ \kappa_\ell \approx \bb
 .
\]
The simplest way to achieve this (and possibly other desirable 
properties) is as follows:
First, for each $\ell$, \change{we consider pairs of}
parameters $\sigma_\ell,\tau_\ell \in \R^2$,
characterized by 
\[
 \br_\ell(\tau_\ell) \approx \bb(\sigma_\ell)
 .
\]
\change{Collecting a sufficiently large, but finite number of 
them yields the set $K_\ell$.}
Such pairs $(\sigma_\ell,\tau_\ell) \in K_\ell$ can be found easily
by requesting that the line connecting the points
$\br_\ell(\tau_\ell)$ and $\bb(\sigma_\ell)$ be perpendicular
to one of the two surfaces.
To address the special role of corners, the set $K_\ell$ 
must 
contain the $\tau$ values $\tau_\ell^0 := [0,0]$ and
$\tau_\ell^1 := [1,0]$. The corresponding $\sigma$-values must 
be \change{consistent} in the sense that the corners
\[
   \sigma_\ell : = \sigma_\ell^0 = \sigma_{\ell-1}^1
\]
are well-defined.
At this and possibly other data points, the function $\kappa_\ell$
should not only approximate but interpolate the prescribed values. 
More precisely, let 
$K'_\ell \subset K_\ell$ be the subsets of pairs intended
for interpolation and $K''_\ell := K_\ell \setminus K'_\ell$ the 
\change{remainders}. In particular, the pairs 
$(\sigma_\ell^0,\tau_\ell^0)$
and $(\sigma_\ell^1,\tau_\ell^1)$ addressing the corners belong
to $K_\ell'$. Now, the conditions to be satisfied by $\kappa_\ell$ are
\begin{align*}
 \kappa_\ell(\sigma_\ell) &= \tau_\ell
 ,\quad 
 (\sigma_\ell,\tau_\ell) \in K_\ell' \\
 \kappa_\ell(\sigma_\ell) &\approx \tau_\ell
 ,\quad 
 (\sigma_\ell,\tau_\ell) \in K_\ell''
 .
\end{align*}
Standard techniques, including all kinds of least squares methods
or smoothing splines can be used to select $\kappa_\ell$ accordingly 
from a given space of splines. If the problem is overconstrained in 
the sense that no solution exists, the spline space must be enlarged
or the conditions must be relaxed. Conversely, if the solution is not 
unique, the spline space must be reduced or more conditions must be 
specified.

The interpolation conditions
$\kappa_\ell(\sigma_\ell) = [0,0]$ and
$\kappa_{\ell-1}(\sigma_\ell) = [1,0]$
for the corners yield
\[
 \gamma_\ell(0) = \gamma_{l-1}(1) = \sigma_\ell
 ,
\]
as requested by \eqref{eq:gamma_loop}. Still, condition 
\eqref{eq:con1} needs to be addressed. When needed, one determines
matrices $T_\ell \in \R^{3\times 2}$ such that 
\[
 T_\ell \approx D\bb(\sigma_\ell)
 ,\quad 
 \bn_\ell(0,0)^{\rm t} \cdot T_\ell = 0
 ,\quad 
 \ell = 1,\dots,L
 .
\]
In particular, the columns of $T_\ell$ span the wanted tangent space
at $\sigma_\ell$. Now, $\kappa_\ell$ is determined as before, but 
such that is satisfies the additional conditions
\begin{align*}
 D\br_\ell(0,0) \cdot D\kappa_\ell(\sigma_\ell) &= T_\ell\\
 D\br_\ell(1,0) \cdot D\kappa_\ell(\sigma_{\ell+1}) &= T_{\ell+1}
 .
\end{align*}
As the chain rule shows, this is in accordance with 
$\br_\ell \circ \kappa_\ell \approx \bb$, and it also
implies \eqref{eq:con1}.

\subsection{Weights}

Once the reparametrizations $\kappa_\ell = [p_\ell,q_\ell]$ are
known, also suitable weights $w$ and $w_\ell$ can be determined
conveniently. 

By definition, the second coordinate $q_\ell$ of $\kappa_\ell$ 
vanishes at $\gamma_\ell$. Assuming that the 
ABC-surface $\ba$ to be determined should have contact order 
$k_\ell \ge 0$ at the boundary segment $\Gamma_\ell$, the
simplest choice for $w$ is
\begin{equation}
\label{eq:w}
 w = q_1^{r_1} \cdots q_L^{r_L}
 ,
\end{equation}
where the exponents are  $r_\ell := k_\ell+1$.
At this point, it has to be made sure that each factor $q_\ell$
is positive on $\Gamma\setminus\Gamma_\ell$.
Additional zeros away from $\Gamma_\ell$ are not critical. Increasing 
some control points of B-splines vanishing at that boundary segment, 
these zeros can be removed without changing the crucial boundary 
behavior. However, zeros near the boundary may be persistent
so that the whole construction fails. In particular, reentrant 
corners are likely to cause such a situation. This case requires 
a special treatment and will be discussed in a forthcoming report.

Omitting one factor yields possible weights
\[
 w_\ell = q_1^{r_1} \cdots 
 q_{\ell-1}^{r_{\ell-1}}
 q_{\ell+1}^{r_{\ell+1}} \cdots q_L^{r_L}
 .
\]
Let us check the conditions concerning contact order given in
Definition~\ref{def:contact}.
For $\sigma \in \Gamma_j$, we obtain
\[
 \bar w (\tau)= \frac{w(\sigma+\tau)}{w_j(\sigma+\tau)} 
 = q_j^{r_j}(\sigma+\tau) = O(|\tau|^{r_j})
\]
since $q_j(\sigma) = 0$. Further, 
\[
 \bar w_\ell (\tau)= \frac{w_\ell(\sigma+\tau)}{w_j(\sigma+\tau)} 
 = \frac{q_j^{r_j}(\sigma+\tau)}{q_\ell^{r_\ell}(\sigma+\tau)} = 
 O(|\tau|^{r_j})
\]
unless the denominator vanishes at $\sigma$, and this is only
possible if $\ell = j$ or $\ell=j-1$ and $\sigma = \sigma_j$.
Exactly these cases are excluded in the definition.

The weights suggested above have two drawbacks: First, the influence 
of each ribbon extends over the whole surface, which might be 
unwanted.
Second, the polynomial degree of such weights is increasing with $L$.
If, for example, $L=6$, $\deg q_\ell = [3,3]$, and $r_\ell=3$
for all $\ell$ to achieve contact order $k=2$ everywhere, then 
$\deg w = \sum_{\ell=1}^L r_\ell\deg q_\ell= [54,54]$.
This might cause problems with efficiency of storage and speed
of evaluation after converting the surface $\ba$
from its procedural definition \eqref{eq:ABC} to the usual form
\[
 \ba = \frac{\sum_i \omega_i b_i \bp_i}{\sum_i \omega_i b_i}
\]
of a NURBS surface, i.e., the linear combination of control points
$\bp_i$ by means of B-splines $b_i$ and (constant) weights 
$\omega_i$. The latter is requested in particular for the export of 
$\ba$ to standard file formats, as addressed in the next section.
The following modification of weights away from the boundary of 
the domain $\Gamma$ solves both issues:

\begin{figure}[t]
\begin{center}
\includegraphics[width=.45\textwidth]{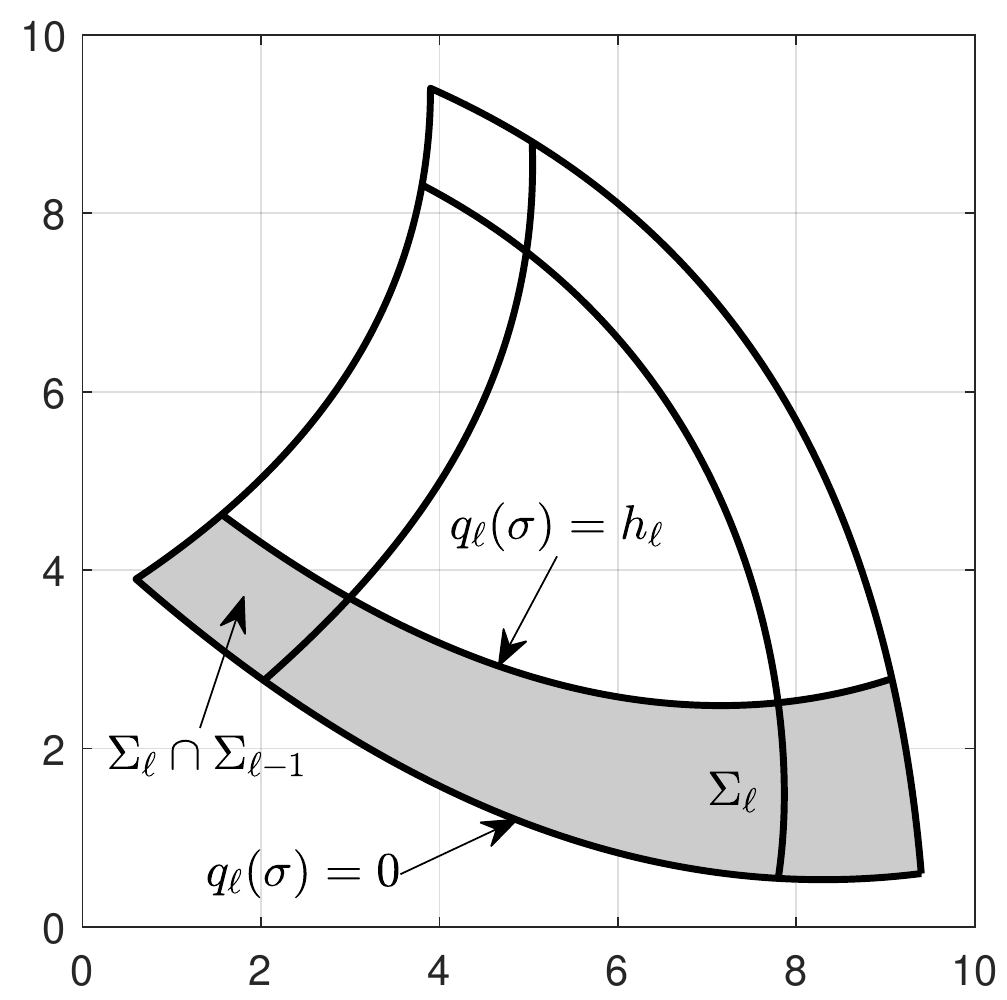}
\qquad 
\includegraphics[width=.45\textwidth]{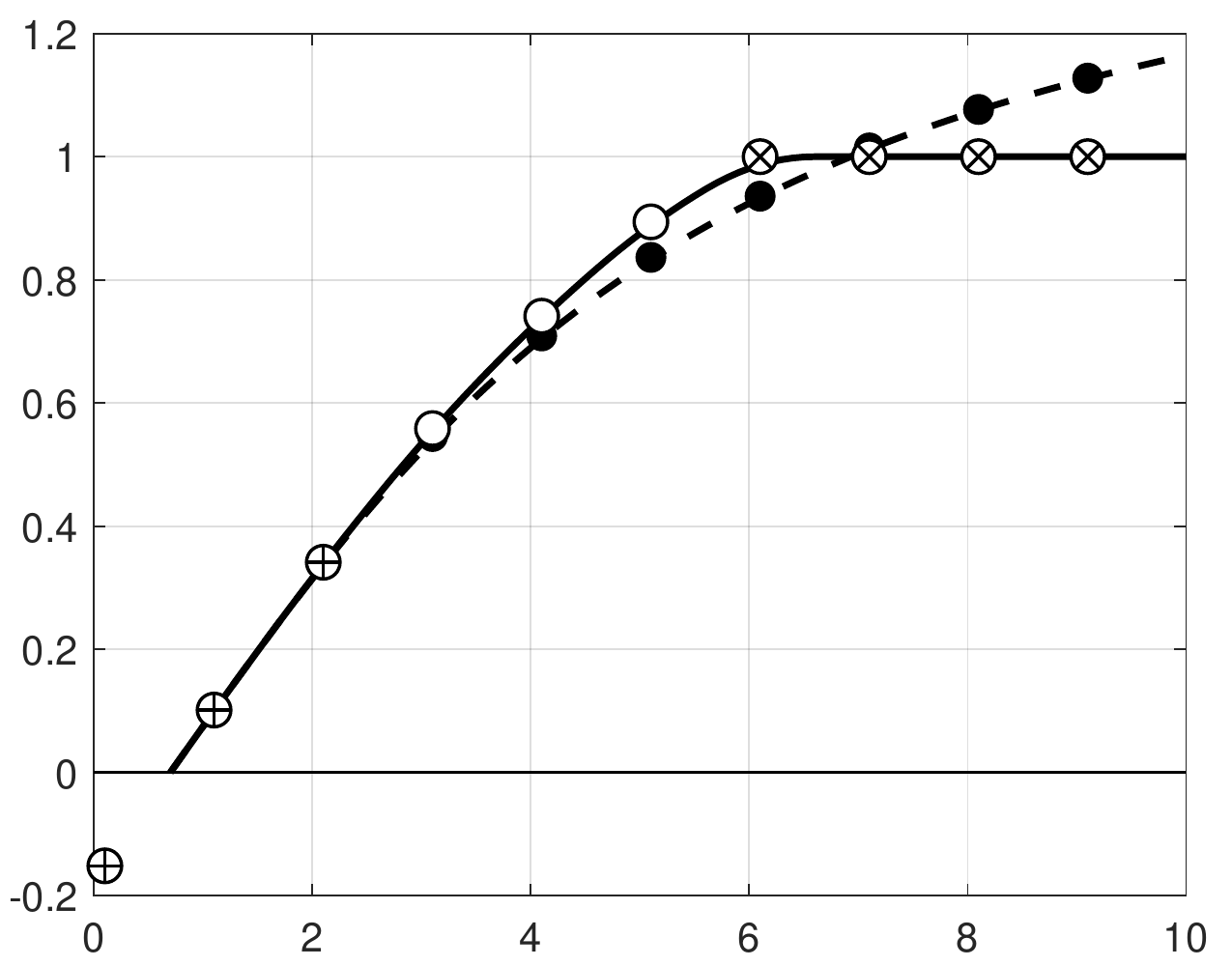}
\end{center}
\caption{Boundary stripes {\it (left)} and univariate plateau 
generation with 
original control points~$\bullet$,
unchanged control points~$\oplus$,
constant control points~$\otimes$,
and optimized control points~$\circ$
{\it (right)}.
}
\label{Stripes+Plateau}
\end{figure}
Denote by
\[
 \Sigma_\ell := \{\sigma \in \Gamma: q_\ell(\sigma) \le h_\ell\}
\]
the {\em boundary stripes} of $\Gamma$ with 
\change{{\em width $h_\ell$}}.
The widths $h_\ell>0$ are assumed to be chosen small enough to make 
sure that only boundary stripes of neighboring curves intersect, and 
that these intersections as well as the stripes themselves consist of 
a single connected piece, see Figure~\ref{Stripes+Plateau} 
{\it (left)}. 
Now, 
we derive from $q_\ell$ a new spline $\bar q_\ell$ as follows:
First, classify the indices in the B-spline representation $q_\ell 
= \sum_i q_\ell^i b_i$ according to 
\begin{align*}
 I_\ell &:= \{i : q_\ell(\sigma)=0 \text{ for some } \sigma \in \supp 
b_i\}
\\
 J_\ell &:= \{i : q_\ell(\sigma)>h_\ell \text{ for some } \sigma \in 
\supp 
b_i\}
 .
\end{align*}
That is, B-splines $b_i$ affect the boundary if $i \in I_\ell$, and 
they affect the part beyond the boundary stripe if $i \in J_\ell$.
Second, if the sets $I_\ell$ and $J_\ell$ are not disjoint, refine 
the B-spline representation of $q_\ell$ by knot insertion and repeat 
the first step.
Third, set 
\begin{equation}
\label{eq:cut}
 \bar q_\ell^i := 
 \begin{cases}
  q_\ell^i & \text{ if } i \in I_\ell\\
  h_\ell & \text{ if } i \in J_\ell.
 \end{cases}
\end{equation}
Coefficients with indices neither in $I_\ell$ nor in $J_\ell$ are 
determined such that the resulting spline has reasonable shape, for 
instance by minimizing a standard fairness functional.
Fourth, define the spline $\bar q_\ell := \sum_{i} \bar q_\ell^i b_i$.
By construction, $\bar q_\ell$ is an implicit representation of 
$\gamma_\ell$, but has a {\em plateau} with constant value $h_\ell$ 
beyond the boundary stripe $\Sigma_\ell$.
Figure~\ref{Stripes+Plateau} {\it (right)} illustrates the procedure
exemplarily in the univariate case.
Replacing \eqref{eq:w} by 
\[
 w = \bar q_1^{r_1} \cdots \bar q_L^{r_L}
 , 
\]
we obtain an implicit representation of $\partial \Gamma$ with 
plateau on the interior part of the domain that 
is not covered by a boundary stripe, see 
Figure~\ref{WeightsPlateau} {\it (left)}.
By construction, the coordinate degree of $w$ is bounded by 
\[
 \deg w \le \max_\ell 
 \bigl(r_\ell \deg q_\ell + r_{\ell+1} \deg q_{\ell+1}\bigr)
\]
since at any given point $\sigma \in \Gamma$ at most two factors are 
non-constant. Notably, the degree does not depend on the number $L$ 
of boundary segments.
In the example above, we have now $\deg q = [18,18]$.
This is still a relatively high degree, but much less than before, 
and can be considered manageable.

\begin{figure}[t]
\begin{center}
\includegraphics[width=.45\textwidth]{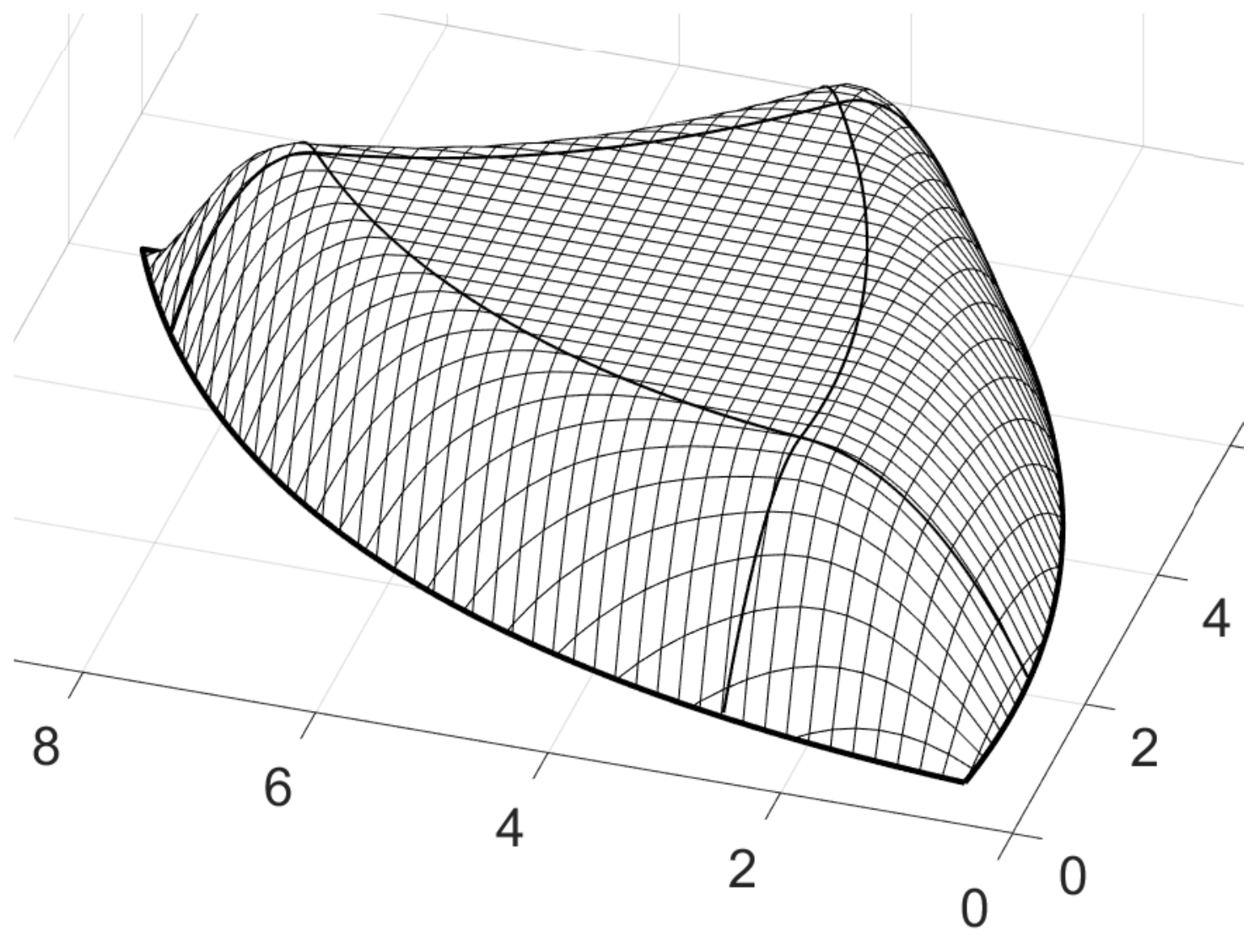}
\qquad 
\includegraphics[width=.45\textwidth]{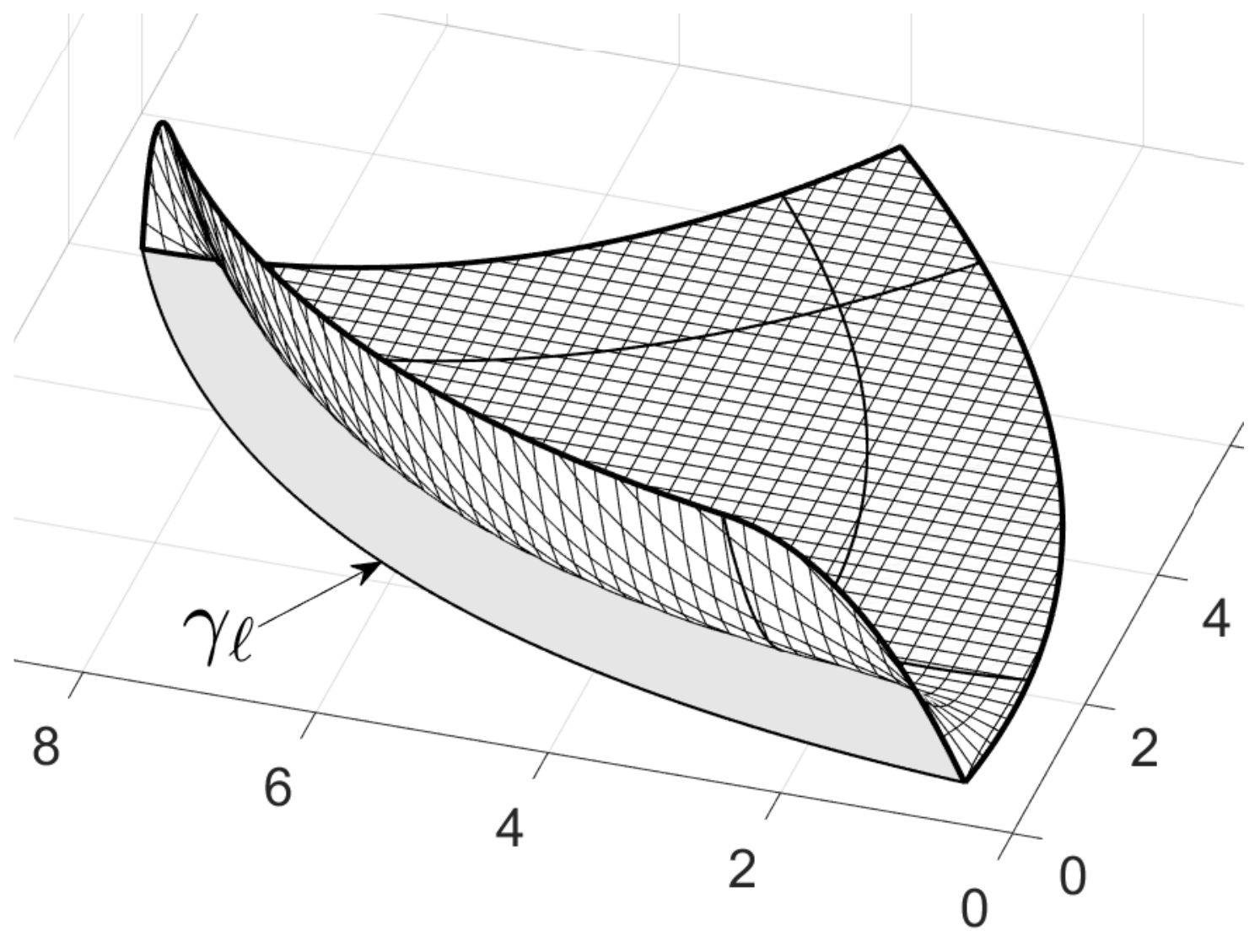}
\end{center}
\caption{Weights $w$ {\it (left)} and $w_\ell$ {\it (right)}
with plateau.}
\label{WeightsPlateau}
\end{figure}
The weights $w_\ell$ may be modified in a similar spirit. 
We define the functions 
\[
 \hat q_\ell := \bar q_{\ell-1} \bar q_{\ell+1}
\]
and process them exactly like the functions $q_\ell$ above,
but now replacing \eqref{eq:cut} by
\[
  \tilde q_\ell^i := 
  \begin{cases}
  \hat q_\ell^i & \text{ if } i \in I_\ell\\
  0& \text{ if } i \in J_\ell.
 \end{cases}
\]
This way, we produce a function with a zero plateau 
beyond the boundary stripe. The desired weight is now simply given by
\[
 w_\ell := \tilde q_\ell^{r_\ell}
 ,
\]
see Figure~\ref{WeightsPlateau} {\it (right)}. The degree
\[
 \deg w_\ell = r_\ell \max(\deg q_{\ell-1}, \deg q_{\ell+1})
\]
is again independent of $L$.

Let us reconsider the degree of $\ba$
for weights $w$ and $w_\ell$ with plateau, as defined above.
For simplicity, we assume that 
$\deg \bb = \deg \kappa_\ell = [n,n]$ and $\deg \br_\ell = [n,m]$,
and that all contact orders are equal to $k = r-1$.
Then, using \eqref{eq:deg},
\[
 \deg \ba = [n,n] \cdot \max\{2r+1,r+n+m\}
 .
\]
The coordinate degree of the ribbons is set to the unmatched pair 
$[n,m]$ for the following reason: In the first direction,
corresponding to the behavior along the boundary, full flexibility
concerning the choice of knots and degrees
is provided to model a large variety of shapes. 
By contrast, in the second 
direction, corresponding to the behavior transversal to the boundary,
it is sufficient to deal with pure polynomials of low degree.
For instance, choosing $m=1$ allows \change{the definition of} the 
boundary curves
$\br_\ell(u,0)$ and the boundary normals $\bn_\ell(u,0)$ along
$\Gamma_\ell$. When also the curvature tensors $E_\ell(u,0)$
shall be prescribed, it suffices to choose $m= 2$.
The minimal configuration yielding a $G^1$-patch and contact 
order $k=1$ requests $n=2$ and $m=1$, leading to $\deg \ba = [10,10]$.
For a $G^2$-patch and contact order $k=2$, one chooses at least 
$n=3$ and $m=2$ to obtain $\deg \ba = [24,24]$.

Modeling with surfaces of such a high degree is commonly considered
to be impractical. However, this is never requested. The geometric 
building blocks--base and ribbons--can be assumed to be splines 
of low degree, and their control points can be manipulated
(manually or automatically) as usual. The conversion of 
the ABC-surface $\ba$ to NURBS form
\[
 \ba = \frac{\sum_i \omega_i b_i \bp_i}{\change{\sum_i}w_i b_i}
 ,
\]
as detailed in the next section,
serves only the purpose to prepare it for file export according 
to industrial standards. It should be considered as a 
{\em read-only representation}, to be used for 
evaluation purposes. Never ever, the control points $\bp_i$ 
should be touched again after being derived from \eqref{eq:ABC}.


\section{File export}

As a matter of fact, data exchange between industrial CAD/CAM systems
is limited to a few completely standardized file formats, like
IGES or STEP. The only type of surfaces commonly available therein are
trimmed NURBS. Other modeling paradigms, say based on box splines, 
exponential splines, or subdivision are typically not supported.
For the trimming of a NURBS surface $\bs : \R^2 \to \R^3$, the 
following three options can be assumed to be available:

{\bf Parametric trimming.} A closed loop of planar spline curves 
$\gamma_1,\dots,\gamma_L$ without self-intersections is specified.
The compact set $\Gamma \subset \R^2$ bounded by these curves is 
used as the domain of the trimmed surface $\bs$.

{\bf Geometric trimming.} A closed loop of spatial spline curves 
$\bc_1,\dots,\bc_L$ is specified. Since it cannot be assumed that 
the traces of these curves are contained in the surface, the 
corresponding planar curves are defined by 
$\gamma_\ell(u) := \alpha$ where $\alpha \in \R^2$ is chosen such 
that $\change{\bs}(\alpha) \approx \bc_\ell(u)$ with the help of some
projection method.
These curves $\gamma_\ell$ need not be splines by themselves,
but are computed numerically. They are merely assumed to form a 
closed loop, which is bounding the domain $\Gamma$.

{\bf Hybrid trimming.} Some segments of the boundary are
defined by planar spline curves, and the other ones by 
spatial spline curves. The planar curves together with 
the preimages of the projected spatial curves are assumed 
to form a closed loop, which is bounding the domain $\Gamma$.

As we have seen, ABC-spline surfaces are piecewise rational.
But nevertheless, in general, they cannot be represented as trimmed 
NURBS in the requested way. The point is that the summands 
$\br_\ell\circ\kappa_\ell$
may produce inner boundaries between rational segments that 
are not axis-aligned \change{in the parameter domain}
and hence cannot be understood as
knot lines. To solve this issue, we present a special 
construction of the reparametrizations $\kappa_\ell$
that leads to a situation where $\ba$ can be represented as
the union of separate parts, each of which is a genuine trimmed
NURBS surface with boundary segments given either 
geometrically by a 3d spline curve or parametrically
by a 2d spline curve.

For a spline $s$ of degree $n$ with \change{knot vector} 
$U = [u^0,\dots,u^N]$, 
the natural domain of definition is the interval 
$[u^n, u^{N-n}]$. We assume that $s$ can be evaluated on all of $\R$
by simply prolonging the leftmost polynomial piece from 
$[u^n,u^{n+1}]$ to $[-\infty,u^{n+1}]$ and the rightmost
piece from $[u^{N-n-1},u^{N-n}]$ to $[u^{N-n-1},\infty)$.
Hence, true break points, where higher order derivatives
are discontinuous, appear only at the {\em inner knots}
$u^{n+1},\dots,u^{N-n-1}$.

Denote the knot vectors of $\br_\ell$ by $U_\ell,V_\ell$.
We start with considering the second vector $V_\ell$.
Critical points $\sigma \subset \R^2$ that are mapped to
the 
\change{knot line $\R \times \{v_\ell^j\}$ for some 
inner knot $v_\ell^j \in V_\ell$} by the 
reparametrization $\kappa_\ell = [p_\ell,q_\ell]$ 
are characterized by
\[
 q_\ell(\sigma) = v_\ell^j
 .
\]
In general, neither the levelsets of $q_\ell$ nor
their images under $\br_\ell$ are splines.  
Special constructions for $q_\ell$ to create such a situation 
are possible, but impractical. It is much easier and 
completely sufficient to request that $V_\ell$ does not possess 
any inner knots at all. For instance, we may assume
$V_\ell = [0,\dots,0,1,\dots,1]$, what is just the knot vector
for polynomials in B\'ezier form. We identified this setting already 
at the end of the preceding section as a legitimate choice.
To sum up,
ribbons that are pure polynomials in the second coordinate
are able to carry information on location, unit normal,
and curvature tensor along the boundary curve, and they 
do not have breaks in the second coordinate.

The treatment of the first variable is more involved since 
assuming also this dependence to be purely polynomial
would unduly restrict the variety of shapes that can be modeled.
Critical points $\sigma \subset \R^2$ that are mapped to
the inner knot $u_\ell^i \in U_\ell$ by 
$\kappa_\ell$ are characterized by
\[
 p_\ell(\sigma) = u_\ell^i
 .
\]
This time, it is appropriate to demand that the critical
levelsets of $p_\ell$ corresponding to inner knots of $U_\ell$
are linear. More precisely, let 
\[
  \mu_\ell^i := \{\sigma\in \Gamma : p_\ell(\sigma) = u_\ell^i 
  \text{ and } 
  w_\ell(\sigma)>0\}
\]
denote the relevant parts of the levelsets of $p_\ell$ corresponding 
to the knots $u_\ell^i$.
In general, these curves are bent,
see Figure~\ref{Levelsets}  {\it (left)}, but
$p_\ell$ can be determined such that
they become straight lines,
see Figure~\ref{Levelsets}  {\it (right)}.
One may proceed as follows:
\begin{figure}[t]
\begin{center}
\includegraphics[width=.45\textwidth]{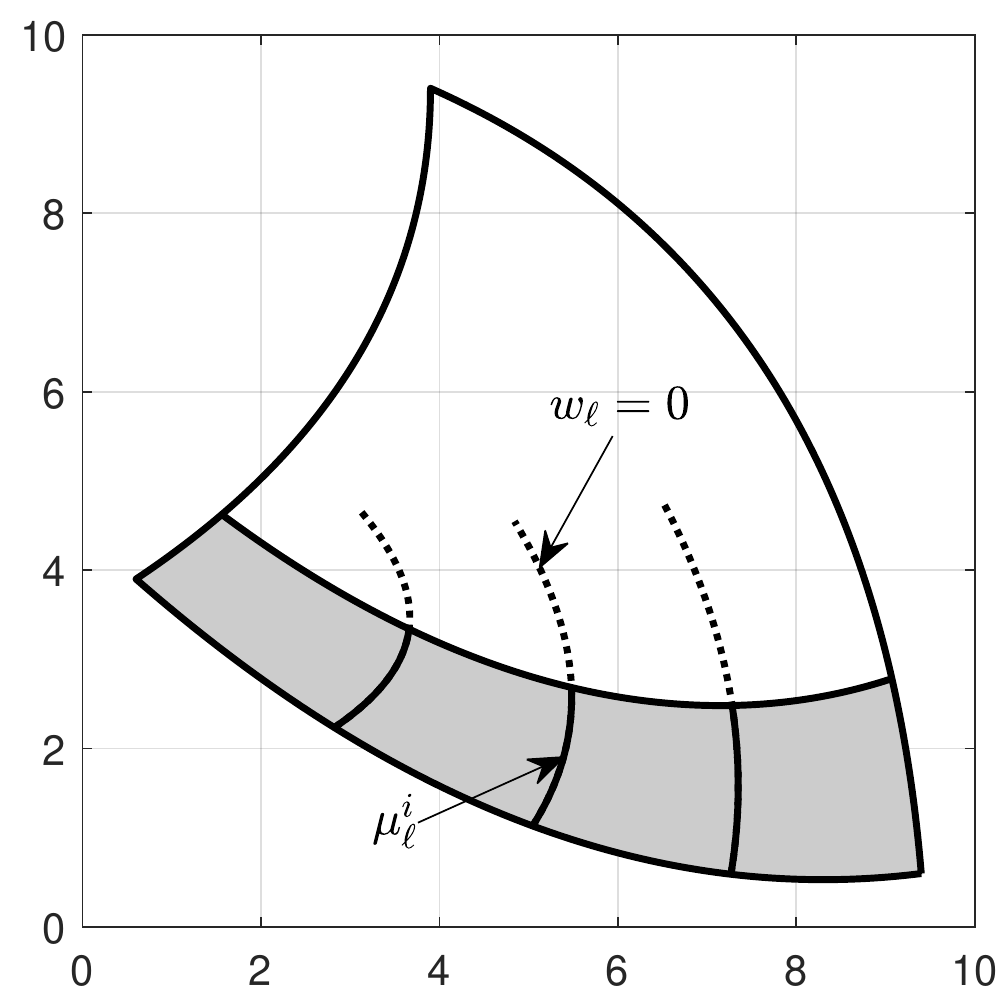}
\qquad 
\includegraphics[width=.45\textwidth]{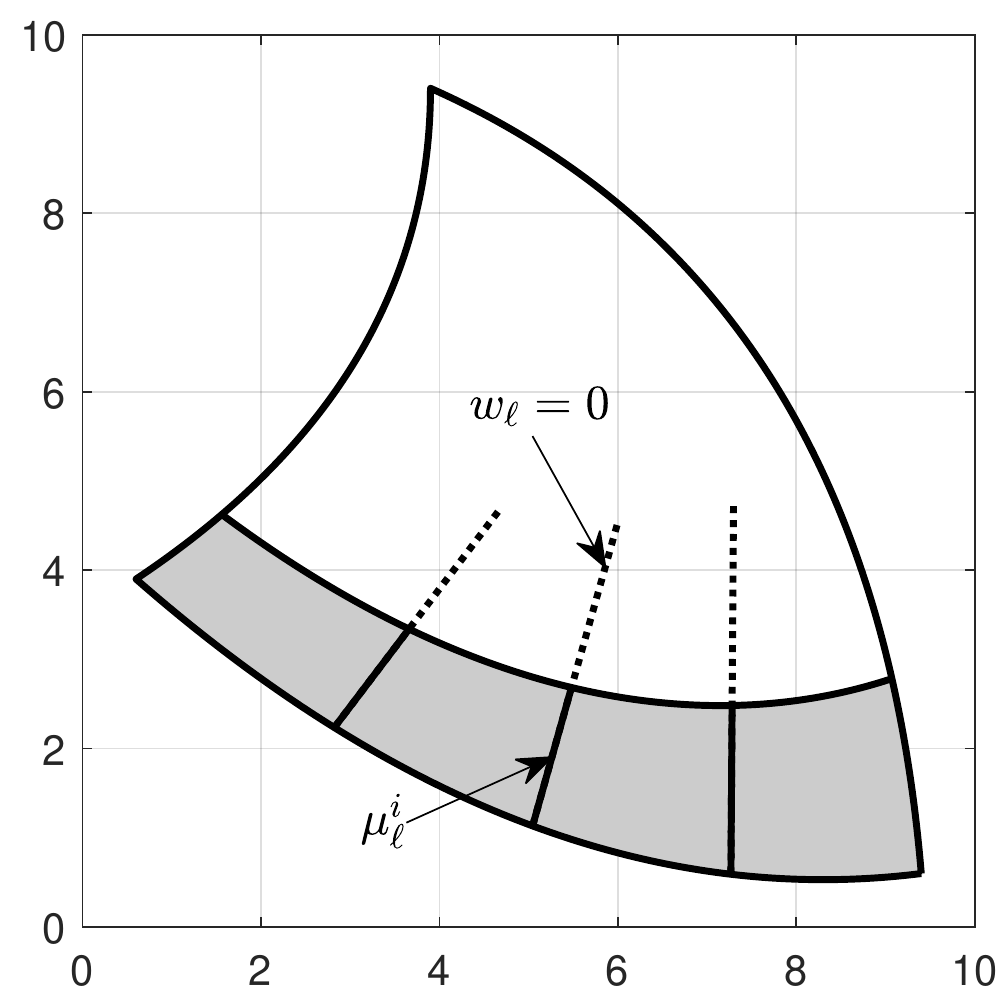}
\end{center}
\caption{Genral {\it (left)} and straight {\it (right)}
levelsets $\mu_\ell^i$. }
\label{Levelsets}
\end{figure}

First,
for each knot $u_\ell^i \in U_\ell$,
choose a direction $\delta_\ell^i \in \R^2$ such that 
\[
 \br_\ell(u_\ell^i,t) \approx 
 \bb(\gamma_\ell(u_\ell^i) + t \delta_\ell^i)
\]
for arguments $\gamma_\ell(u_\ell^i) + t \delta_\ell^i$
within the boundary stripe $\Sigma_\ell$.
Denoting by $D\bb^+ := (D\bb^T D\bb)^{-1} D\bb^T$
the pseudo-inverse of the Jacobian $D\bb$, 
one may set in particular
\[
 \delta_\ell^i := D\bb^+(\gamma_\ell(u_\ell^i)) \cdot 
 \partial_2 \br_\ell(u_\ell^i,0)
 .
\]
Second, define the line segment 
\[
 \mu_\ell^i := \bigl\{ \gamma_\ell^i + t \delta_\ell^i : t \in 
\R\bigr\}
 \cap \Sigma_\ell
\]
with direction $\delta_\ell^i$ within the boundary 
stripe $\Sigma_\ell$. 
Third, 
choose sufficiently many pairs of points $(\sigma_\ell,\tau_\ell)$
of the form 
\[
 \sigma_\ell = \gamma_\ell(u_\ell^i) + t \delta_\ell^i
 \in \mu_\ell^i
 ,\quad 
 \tau_\ell = (u_\ell^i,t)
\]
and add them to the set $K_\ell'$ of data to be interpolated.

Using this setup, points $\sigma \in \Gamma$ that are
mapped to the inner knot $u_\ell^i$ of $\br_\ell$ by $p_\ell$ either 
lie on the straight line $\mu_\ell^i$ or are irrelevant since
$w_\ell(\sigma) = 0$.

Assuming that the ribbons $\br_\ell$ are polynomials in the 
second variable and that the level sets $\mu_\ell^i$ are straight, 
as described, we can 
partition the domain $\Gamma$ into an interior part $\Gamma_0$ and 
boundary parts $\Gamma_\ell^i$
by a network of curves, consisting of three different types:
First, there are the segments $\gamma_\ell^i$ of the trimming 
curves $\gamma_\ell$ corresponding to intervals between 
consecutive knots in $U_\ell$.
Second, there are the level sets $\mu_\ell^i$, and third, 
there are auxiliary spline curves $\alpha_\ell^i$, to be chosen at 
will, connecting the endpoints of the level sets. 
Figure~\ref{Partition} illustrates the 
setting.
Restricting the ABC-surface $\ba$ to the subdomains 
$\Gamma_0,\Gamma_\ell^i$ 
yields pieces $\ba_0,\ba_\ell^i$, respectively, which are trimmed 
NURBS surfaces ready for filing with 
standard data formats. They do not have interior breaks aside from a
regular knot grid, and their boundaries are formed 
by spatial spline curves $\ba \circ \gamma_\ell^i$ and planar spline 
curves $\mu_\ell^i,\alpha_\ell^i$.
\begin{figure}[t]
\begin{center}
\includegraphics[width=.55\textwidth]{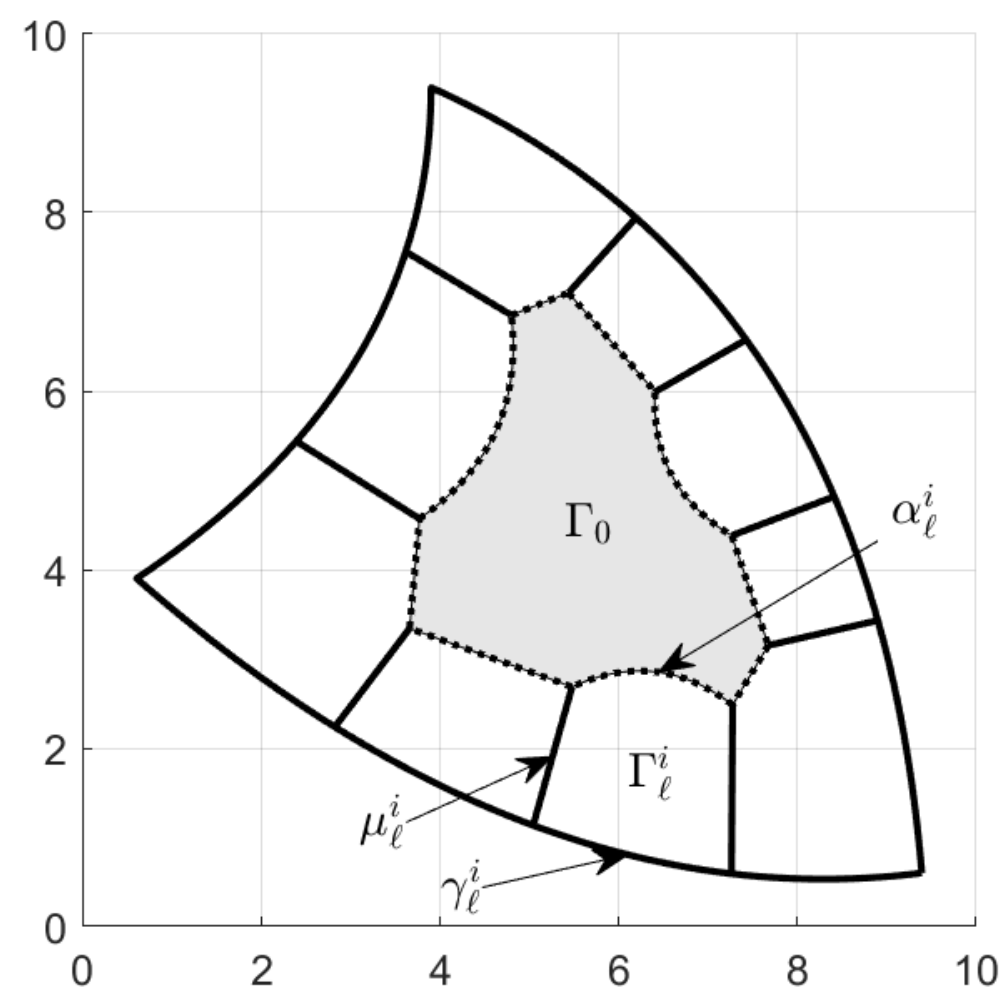}
\end{center}
\caption{Partition of the domain $\Gamma$ into pieces
ready for file export.}
\label{Partition}
\end{figure}


\section{Examples}

The focus of this article is on theory, but still, skipping many
technical details, we 
want to present three examples to illustrate the potential of
our approach. They show in turn composite surfaces with $G^0,G^1$, 
and $G^2$ contact according to theorems~\ref{thm:G0}, \ref{thm:G1},
and \ref{thm:G2}, respectively.
All surfaces, reparametrizations, and weights
are bicubic tensor product splines; all weights are used in 
the version with plateau.

\begin{figure}[t]
\begin{center}
\includegraphics[width=.45\textwidth]{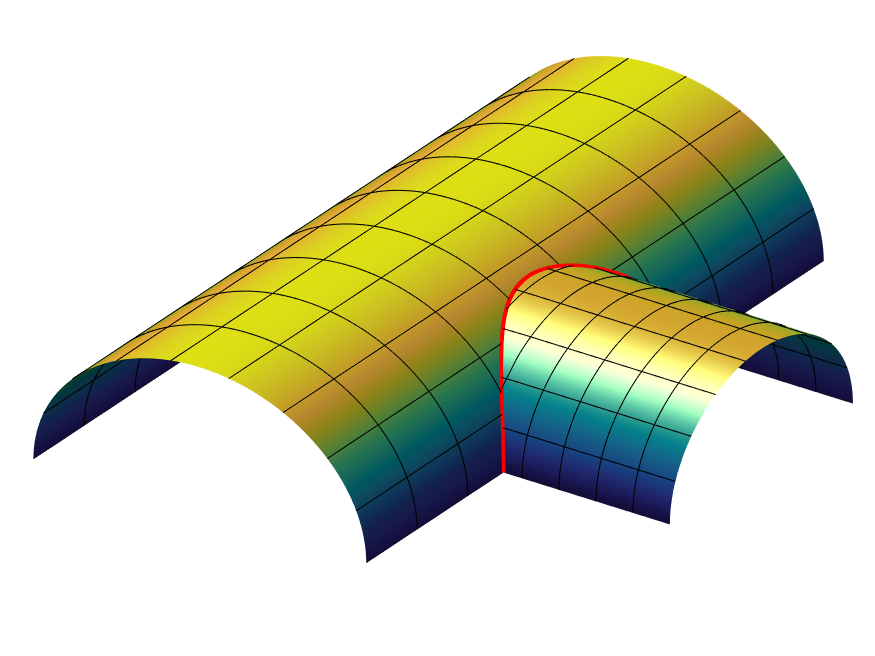}
\hspace{15mm}
\raisebox{7mm}{\includegraphics[width=.35\textwidth]
{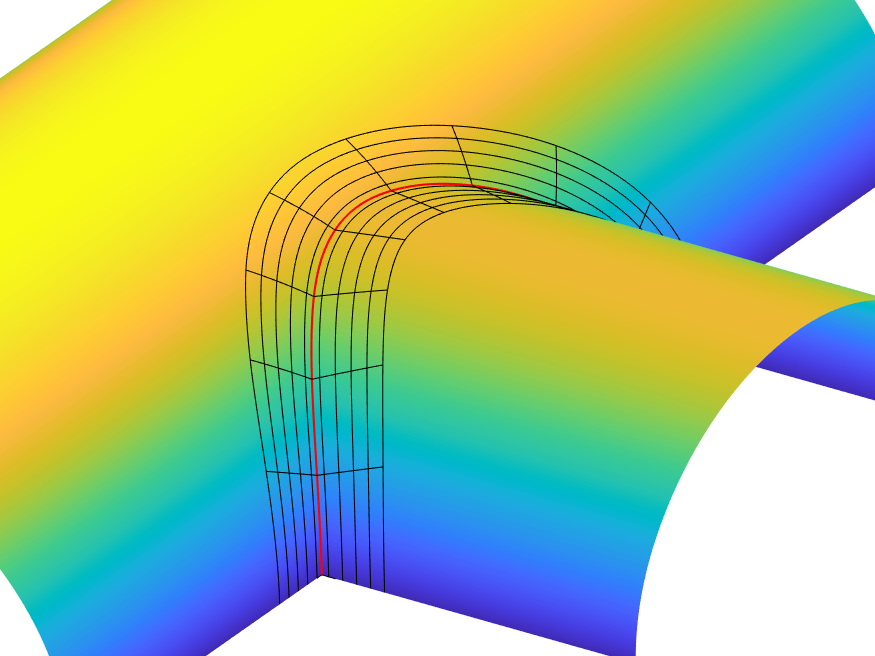}}\phantom{XX}\\
\includegraphics[width=.45\textwidth]{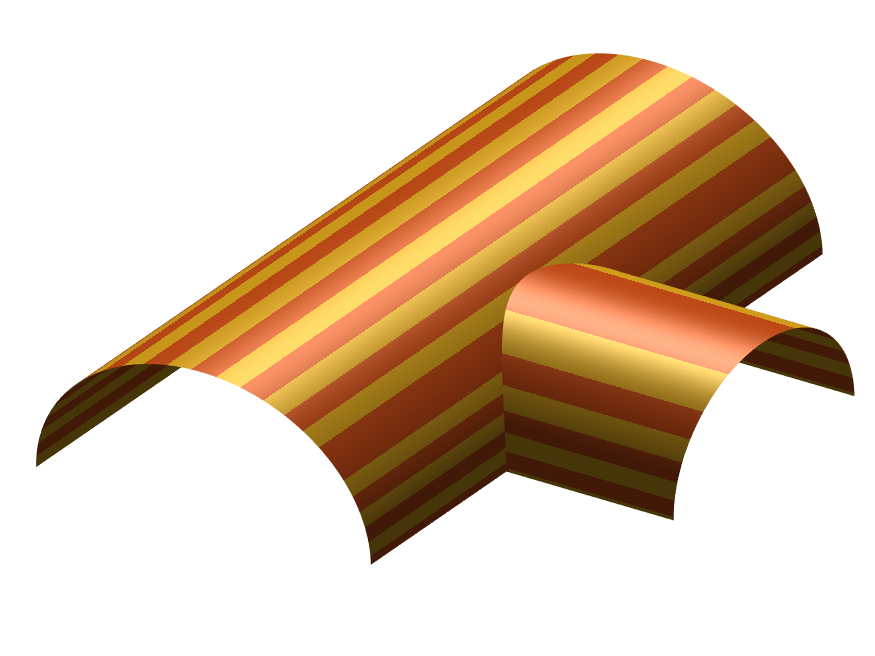}
\qquad 
\includegraphics[width=.45\textwidth]{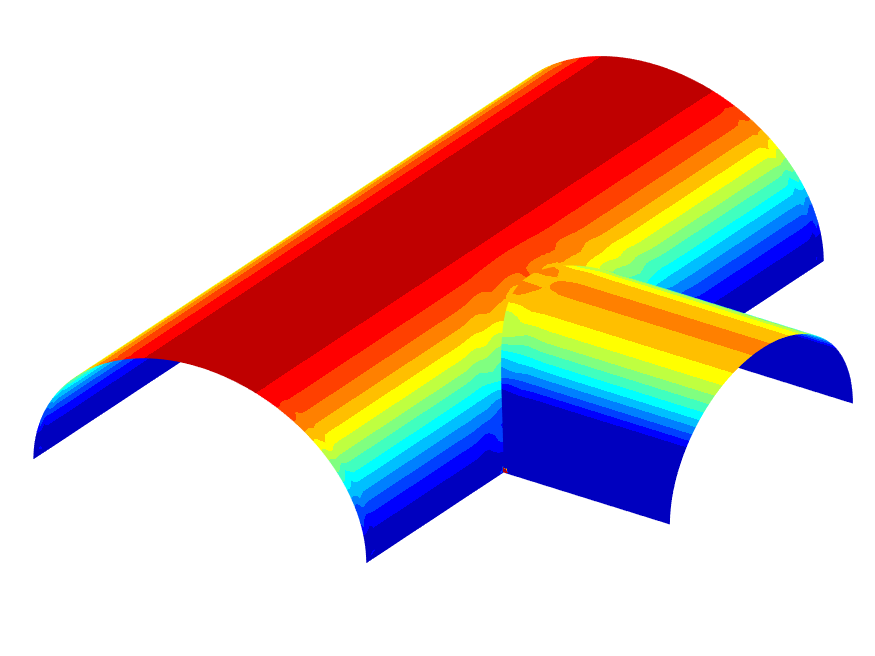}
\end{center}
\caption{Two intersecting cylinders \change{ -- 
paramatric lines {\it (top left)}, 
ribbons near intersection curve {\it (top right)},
isophotes {\it (bottom left)}, 
and mean curvature shading {\it (bottom right)}.
}}

\label{cylinder}
\end{figure}
Figure~\ref{cylinder} presents a typical example from solid modeling:
two intersecting cylinders that trim each other. 
The first plot
shows two ABC-surfaces $\ba^1,\ba^2$ together with a few parameter 
lines, which 
follow the given shapes in the natural way.
The spline representations $\bb^1,\bb^2$ of the two original 
cylinders serve as bases for $\ba^1,\ba^2$, respectively.
The second plot shows the pair of ribbons $\br^1,\br^2$ 
used to model the region
near the intersection curve. Both ribbons have $9$ interior
knots in longitudinal direction.
Their parametrizations are adapted
to the local situation and chosen such that they share a common
curve $\bc = \br^1(\cdot,0) = \br^2(\cdot,0)$.
This curve approximates the exact intersection $\bb^1 \cap \bb^2$
of the 
two original cylinders with high accuracy. The watertight
connection of ribbons is easy to achieve by coinciding 
boundary control points
since both ribbons are in B\'ezier form in transversal direction.
Theorem~\ref{thm:G0} says that this property is inherited 
by the two ABC-surfaces, i.e., they also join seemlessly
at $\bc$.
The third plot shows the impeccable pattern of isophotes, while the 
fourth one is colored by mean curvature.
Only here, one can see that the two surfaces $\ba^1,\ba^2$ slightly
deviate from the perfect cylindrical shape of the original
surfaces. This deviation could be made arbitrarily small 
by choosing finer knot sequences. However, as a matter of fact, it 
is impossible to do without any deviation since the exact intersection
of the original cylinders cannot be represented by a spline
curve--neither in geometry space nor in the parametric domain.

\begin{figure}[t]
\begin{center}
\includegraphics[width=.49\textwidth]{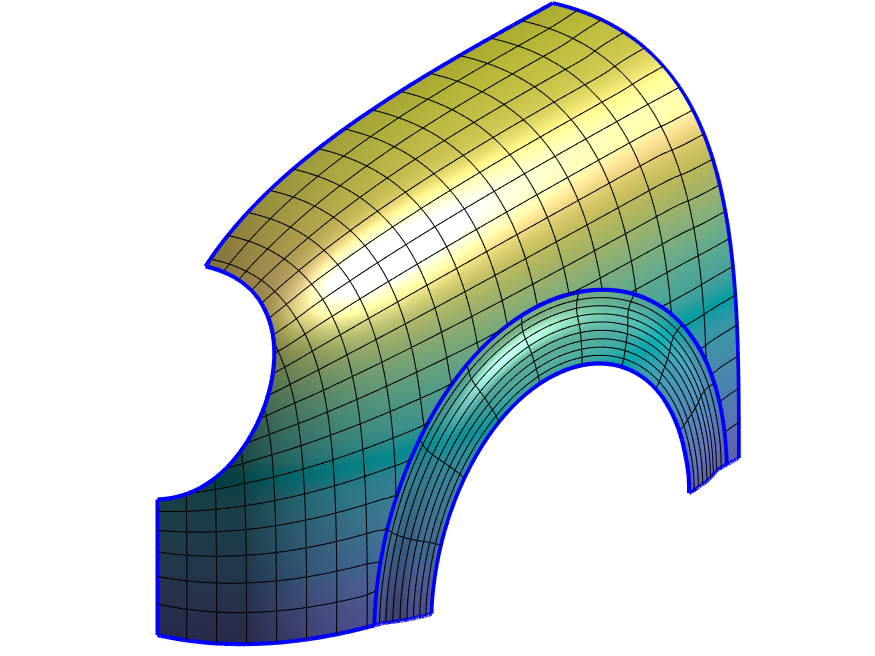}
\includegraphics[width=.49\textwidth]{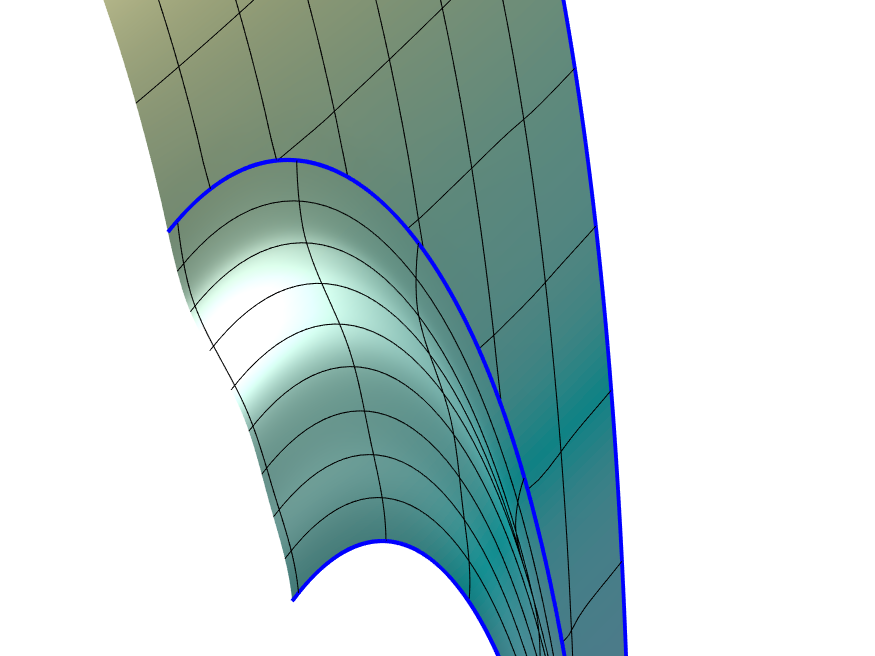}

\medskip

\includegraphics[width=.49\textwidth]{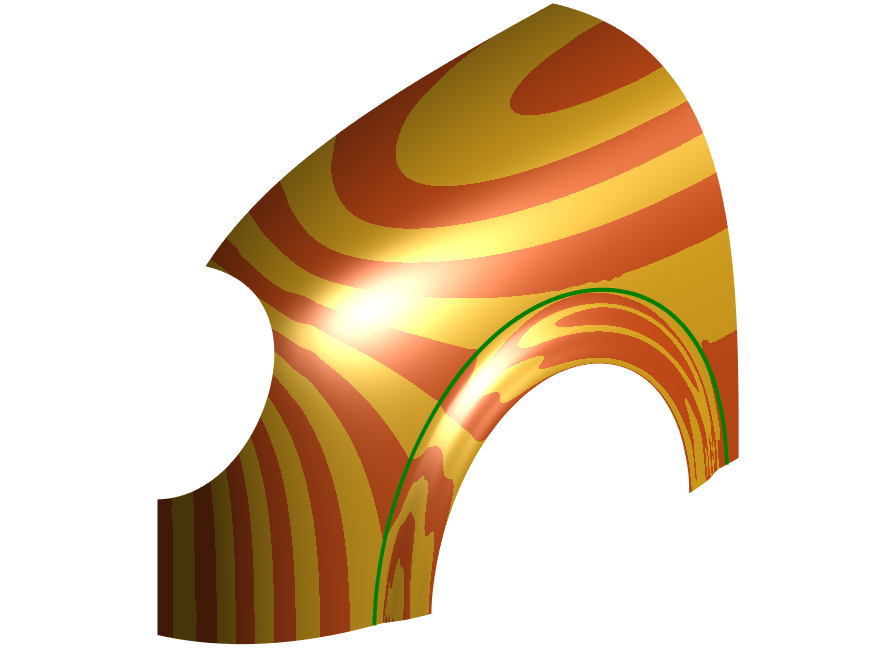}
\includegraphics[width=.49\textwidth]{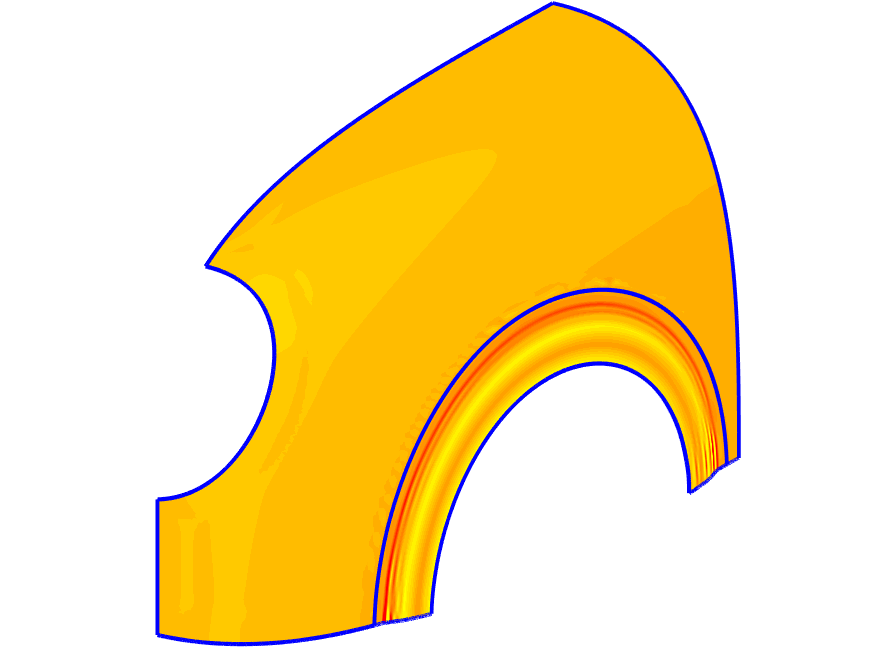}
\end{center}
\caption{Fender of a fictional car \change{ --
whole composite surface {\it (top left)},
zoom on cross section {\it (top right)},
isophotes {\it (bottom left)},
and mean curvature shading {\it (bottom right)}.
}}
\label{car}
\end{figure}
Figure~\ref{car} shows the fender of a fictional car%
\footnote{The geometry shown here was kindly provided 
my Malcolm Sabin, who modeled a similar composite $G^2$-surface
using a completely different, hitherto unpublished technique.}.
The whole geometry consists of a bigger, uniformly curved
part and a bulged-out, annular
area that can be viewed as part of the wheel case.
While it is possible to model the object by a 
single spline surface, it is much more convenient and 
natural to assemble it from two pieces. The ring-shaped
part is modelled by an untrimmed standard spline surface $\bs$.
The other part is an ABC-surface $\ba$, whose ribbon along 
the junction has $C^1$-contact with $\bs$, which is easy to achieve.
This implies coinciding normals along the common boundary,
and hence, by Theorem~\ref{thm:G1}, a $G^1$ contact between 
$\ba$ and $\bs$.
The first plot shows the complete composite surface, while 
the second one features \change{a zoom on} a cross-section profile.
The third and fourth plot show isophotes (which are continuos) and 
mean curvature (which is discontinuous), respectively.

\begin{figure}[t]
\begin{center}
\includegraphics[width=.45\textwidth]{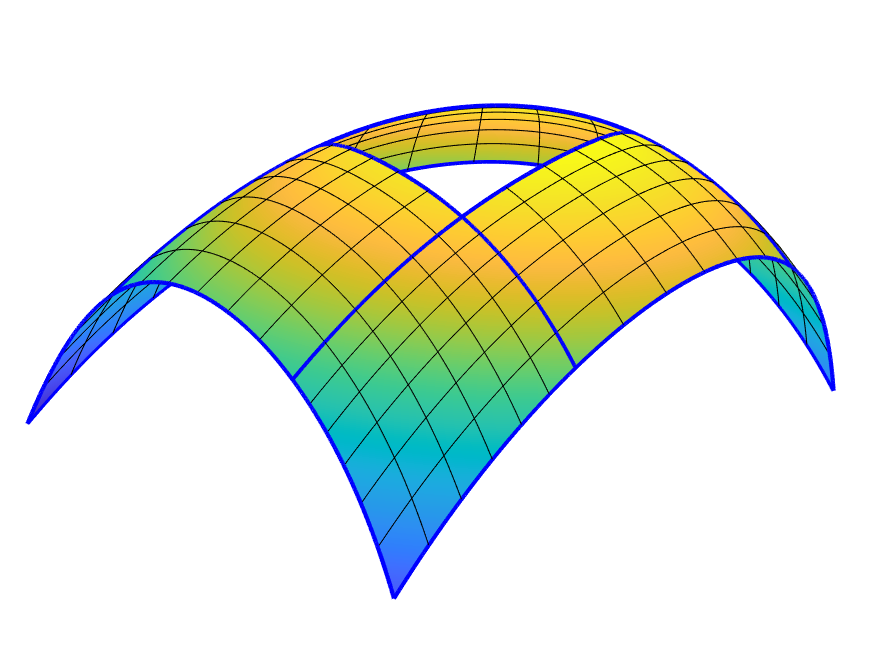}
\qquad
\includegraphics[width=.45\textwidth]{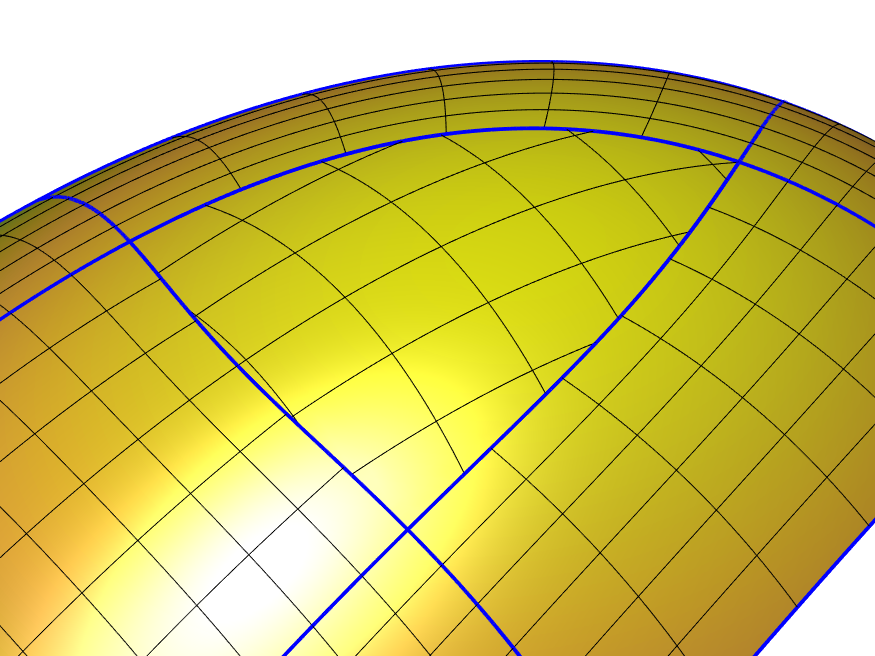}

\includegraphics[width=.45\textwidth]{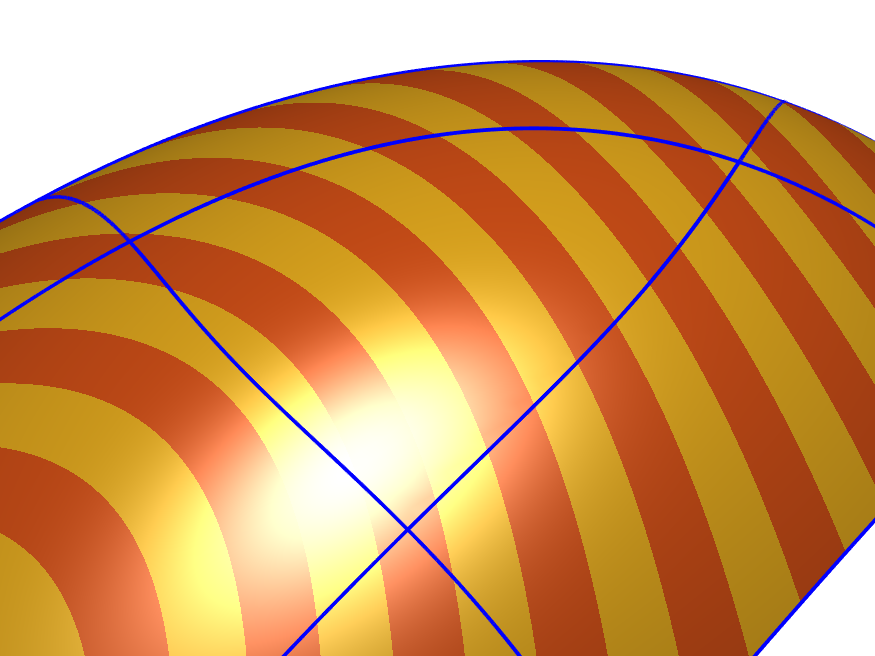}
\qquad 
\includegraphics[width=.45\textwidth]{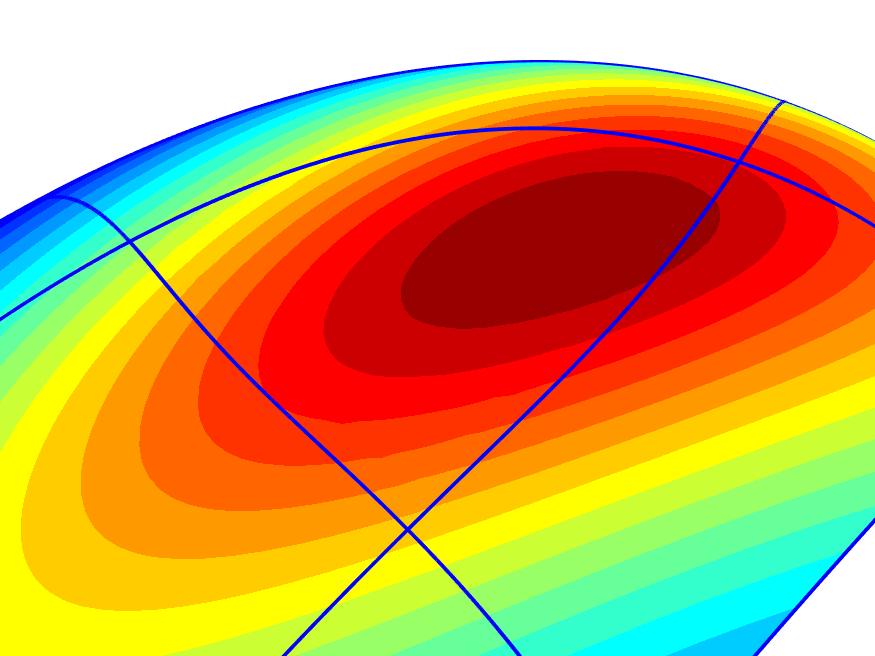}
\end{center}
\caption{Filling of a triangular hole \change{ --
Arrangement of standard patches surrounding the hole {\it (top left)},
surrounding patches and ABC-filling with some parametric lines
{\it (top right)},
isophotes {\it (bottom left)},
and mean curvature shading {\it (bottom right)}.
}}
\label{triangle}
\end{figure}
Figure~\ref{triangle} presents a classical hole-filling 
problem. Several standard spline surfaces are surrounding
a triangular region, and another surface is sought that
closes the gap smoothly. Here, we show that this can 
be achieved by an ABC-surface with three ribbons,
each joining $C^2$ with the corresponding outer part.
By Theorem~\ref{thm:G2}, this guarantees a $G^2$-contact.
The first plot shows the configuration defining the 
three-sided hole, and the second one the completed 
structure together with a few parameter lines.
Isophotes and also the distribution of
mean curvature, as shown by the third and fourth plot,
indicate that the generated shape leaves little 
room for improvement.

\bigskip
\noindent
{\bf Acknowledgement.} We would like to thank Malcolm Sabin 
for fruitful discussions and comments.

\bibliographystyle{alpha}
\bibliography{ref}
\end{document}